\documentclass{sig-alternate-05-2015}

\usepackage{amsmath,amssymb,amsfonts} 

\usepackage{algorithm}
\usepackage{algorithmic}
\usepackage{xcolor}

\newtheorem{problem}{Problem}
\newtheorem{theorem}{Theorem}

\newtheorem{lemma}[theorem]{Lemma}
\newtheorem{proposition}[theorem]{Proposition}

\newtheorem{example}{Example}

\usepackage{hyperref}
\hypersetup{colorlinks=true,urlcolor=blue,linkcolor=blue,citecolor=blue}
\usepackage{cleveref}

\usepackage{mathtools}

\newcommand{\USRIM}{\sf {US-RIM}}
\newcommand{\USRIMA}{\sf {US-RIM-A}}
\newcommand{\USRIMM}{\sf {US-RIM-M}}
\newcommand{\ICSRIM}{{\sf ICS-RIM}}
\newcommand{\OESRIM}{{\sf OES-RIM}}

\newcommand{\abs}[1]{\left| #1 \right|} 
\newcommand{\norm}[1]{\left\| #1 \right\|} 
\newcommand{\supp}{\textnormal{supp}} 
\newcommand{\E}{\mathbb{E}} 
\newcommand{\A}{\mathcal{A}} 
\newcommand{\greedy}{{\sf Greedy}} 
\newcommand{\lugreedy}{{\sf LUGreedy}}  
\newcommand{\lu}{\mathsf{LU}}

\newcommand{\NP}{{\bf NP}}

\DeclareMathOperator*{\argmax}{arg\,max}

\newcommand{\compilehidecomments}{false}
\ifthenelse{ \equal{\compilehidecomments}{true} }{%
	\newcommand{\wei}[1]{}
	\newcommand{\tian}[1]{}
	\newcommand{\zihan}[1]{}
	\newcommand{\mingfei}[1]{}
	\newcommand{\xuren}[1]{}
}{
\newcommand{\wei}[1]{{\color{blue}  [\text{Wei:} #1]}}
\newcommand{\tian}[1]{{\color{gray}  [\text{Tian:} #1]}}
\newcommand{\zihan}[1]{{\color{brown}  [\text{Zihan:} #1]}}
\newcommand{\mingfei}[1]{{\color{purple}  [\text{Mingfei:} #1]}}
\newcommand{\xuren}[1]{{\color{green}  [\text{Xuren:} #1]}}
}

\newcommand{\submissiontype}{kdd} 

\newcommand{\texkdd}[1]{
\ifthenelse{\equal{\submissiontype}{kdd}}{#1}{}%
}

\newcommand{\texarxiv}[1]{
\ifthenelse{\equal{\submissiontype}{arxiv}}{#1}{}%
}

\renewcommand{\submissiontype}{arxiv}

\makeatletter
\def\@copyrightspace{\relax}
\makeatother

\begin{document}

\setcopyright{acmcopyright}

\doi{http://dx.doi.org/10.1145/2939672.2939745}

\isbn{978-1-4503-4232-2/16/08}

\conferenceinfo{KDD '16}{August 13--17, San Francisco, CA, USA}

\acmPrice{\$15.00}

%

\title{Robust Influence Maximization}

%
%
%
%
%

\numberofauthors{5} 
%
\author{
\alignauthor Wei Chen \\
\affaddr{Microsoft Research}\\
\email{weic@microsoft.com}
\and
\alignauthor Tian Lin \\
\affaddr{Tsinghua University}\\
\email{lint10@mails.tsinghua.edu.cn}
\and
\alignauthor Zihan Tan \\
\affaddr{IIIS, Tsinghua University}\\
\email{zihantan1993@gmail.com}
\and
\alignauthor Mingfei Zhao \\
\affaddr{IIIS, Tsinghua University}\\
\email{mingfeizhao@hotmail.com}
\and
\alignauthor Xuren Zhou \\
\affaddr{The Hong Kong University of Science and Technology}\\
\email{xzhouap@cse.ust.hk}
}

\maketitle

\sloppy

\begin{abstract}
	In this paper, we address the important issue of uncertainty in the edge influence
	probability estimates for the well studied influence maximization problem
	--- the task of finding $k$ seed nodes in a social network to maximize
	the influence spread.
	We propose the problem of robust influence maximization, which maximizes
	the worst-case ratio between the influence spread of the chosen seed set and
	the optimal seed set, given the uncertainty of the parameter input.
	We design an algorithm that solves this problem with a solution-dependent bound.
	We further study uniform sampling and adaptive sampling methods to
	effectively reduce the
	uncertainty on parameters and improve the robustness of
	the influence maximization task.
	Our empirical results show that parameter uncertainty may greatly affect influence
	maximization performance and prior studies that learned influence probabilities
	could lead to poor performance in robust influence maximization due to
	relatively large uncertainty in parameter estimates,
	and information cascade based adaptive sampling
	method may be an effective way to improve the robustness of influence maximization.
\end{abstract}

\printccsdesc



\section{Introduction} \label{sect:intro}
In social and economic networks, {\em Influence Maximization} problem has been extensively studied
over the past decade, due to its wide applications to viral marketing \cite{domingos2001mining,kempe2003maximizing}, outbreak detection \cite{leskovec2007cost},
rumor monitoring \cite{budak2011limiting}, etc.
For example, a company may conduct a promotion campaign in social networks
by sending free samples to the initial users (termed as seeds),
and via the word-of-mouth (WoM) effect, more and more users are influenced by social links to join the campaign
and propagate messages of the promotion.
This problem is first introduced by Kempe et al.~\cite{kempe2003maximizing} under an algorithmic framework to find the most influential seeds,
and they propose the {\em independent cascade} model and {\em linear threshold} model,
which consider the social-psychological factors of information diffusion to simulate such a random process of adoptions.

Since Kempe et al.'s seminal work, extensive researches have been done on influence
maximization, especially on improving the efficiency of influence maximization
in the independent cascade model~\cite{chen2009efficient, chen2010scalable, goyal2011celf,borgs14,tang14}, all of which
assume that the ground-truth influence probabilities on edges are exactly known.
Separately, a number of studies \cite{saito2008prediction,tang2009social,goyal2010learning,gomez2011uncover,Netrapalli12} propose learning methods to extract edge influence probabilities.
Due to inherent data limitation, no learning method could recover the
exact values of the edge probabilities, and what can be achieved is the estimates
on the true edge probabilities, with confidence intervals indicating that
the true values are within the confidence intervals with high probability.
The uncertainty in edge probability estimates, however, may adversely affect
the performance of the influence maximization task, but this topic has
left mostly unexplored.
The only attempt addressing this question is a recent study in~\cite{he2015stability},
but due to a technical issue as explained in~\cite{he2015stability},
the results achieved by the study is rather limited.

In this paper, we utilize the concept of robust optimization~\cite{ben2002robust}
in operation research to address the issue of influence maximization
with uncertainty.
In particular, we consider that the input to the influence maximization task is
no longer edge influence probability on every edge of a social graph, but instead
an interval in which the true probability may lie.
Thus the input is actually a parameter space $\Theta$, which is the product of
all intervals on all edges.
For any seed set $S$, let $\sigma_\theta(S)$ denote the {\em influence spread} of $S$
under parameter setting $\theta\in \Theta$.
Then we define {\em robust ratio} of $S$ as
$g(\Theta,S) = \min_{\theta\in \Theta} \frac{\sigma_{\theta}(S)}{\sigma_{\theta}(S^{*}_{\theta})} $,
where $S^{*}_{\theta}$ is the optimal seed set achieving the maximum influence
spread under parameter $\theta$.
Intuitively, robust ratio of $S$ indicates the (multiplicative) gap between
its influence spread and the optimal influence spread under the worse-case
parameter $\theta \in \Theta$, since we are unsure which $\theta \in \Theta$
is the true probability setting.
Then our optimization task is to find a seed set of size $k$ that
maximize the robust ratio under the known parameter space $\Theta$
--- we call this task {\em Robust Influence Maximization (RIM)}.

It is clear that when there is no uncertainty on edge probabilities, which means
$\Theta$ collapses to the single true parameter $\theta$,
RIM degenerates to the classical influence maximization problem.
However, when uncertainty exists, solving RIM may be a more difficult task.
In this paper, we first propose an algorithm {\lugreedy} that
solves the RIM task with a solution-dependent bound on
its performance, which means
that one can verify its performance after it selects the seed
set (Section~\ref{sect:rim}).
We then show that if the input parameter space $\Theta$ is only given and cannot
be improved, it is possible that even the best robust ratio
in certain graph instances could be very small (e.g. $O(\log n / \sqrt{n})$
with $n$ being the number of nodes in the graph).
This motivates us to study sampling methods to further tighten parameter
space $\Theta$, and thus improving the robustness of our algorithm
(Section~\ref{sect:sample}).
In particular, we study both uniform sampling and adaptive sampling for
improving RIM performance.
For uniform sampling, we provide theoretical results on the sample complexity
for achieving a given robust ratio of the output seed set.
For adaptive sampling, we propose an information cascade based sampling heuristic
to adaptively bias our sampling effort to important edges often traversed
by information cascades.
Through extensive empirical evaluations (Section~\ref{sect:experiments}), we show that
(a) robust ratio is sensitive to the width of the confidence interval, and
it decreases rapidly when the width of the confidence interval increases; as
a result prior studies that learned edge probabilities may result in poor robust
ratio due to relative large confidence intervals (and thus high uncertainty);
(b) information cascade based adaptive sampling method performs better than
uniform sampling and other baseline sampling methods, and can significantly
improve the robustness of the influence maximization task.

In summary, the contribution of our paper includes: (a) proposing the problem
of robust influence maximization to address the important issue of
uncertainty in parameter estimates adversely impacting
the influence maximization task;
(b) providing the {\lugreedy} algorithm that guarantees a solution-dependent
bound; and
(c) studying uniform and adaptive sampling methods to
improve robust influence maximization.

\texarxiv{Note that proofs of some technical results can be found in the appendix.}

\texkdd{Due to space constraint, the proofs of some technical results are omitted.
The complete proofs of all results can be found in the full technical report~\cite{ChenLTZZ16}.}

\subsection{Additional Related Work} \label{sect:related}

Influence maximization has been extensively studied and we already point out
a number of closely related studies to our work in the introduction.
For a comprehensive survey, one can refer to the monograph~\cite{chen2013information}.
We discuss a few most relevant work in more detail here.

To the best of our knowledge, the study by He and Kempe~\cite{he2015stability}
is the only attempt prior to our work that also tries to address the
issue of uncertainty
of parameter estimates impacting the influence maximization tasks.
However, besides the similarity in motivation, the technical treatments are quite
different.
First, their central problem, called influence difference maximization, is to
find a seed set of size $k$ that maximizes the additive difference between
the two influence spreads of the {\em same} seed set using different parameter values.
Their purpose is to see how large the influence gap could be due to the
uncertainty in parameter space.
However, our goal is still to find the best possible seed set for influence
maximization purpose, while considering the adverse effect of the uncertainty,
and thus we utilize the robust optimization concept and use the worse-case
multiplicative ratio between the influence spread of the chosen seed set and
the optimal seed set as our objective function.
Second, their influence difference maximization turns out to be hard to approximate
to any reasonable ratio, while we provide an actual algorithm for robust
influence maximization that has both a theoretical solution-dependent bound and
performs reasonably well in experiments.
Third, we further consider using sampling methods to improve RIM, which is not
discussed in~\cite{he2015stability}.

In the context of robust optimization, Krause et al.'s work on robust
submodular optimization~\cite{krause2008robust} is possibly the closest to ours.
Our RIM problem can be viewed as a specific instance of robust submodular
optimization studied in~\cite{krause2008robust}.
However, due to the generality of problem scope studied in~\cite{krause2008robust},
they show strong hardness results and then they have to resolve to
bi-criteria solutions.
Instead, we are working on a particular instance of robust submodular optimization,
and their bi-criteria solution may greatly enlarge the selected seed set size,
which may not be allowed in our case.
Furthermore, they work on finite set of submodular functions, but in our case
our objective function is parametrized with $\theta$ from a continuous
parameter space $\Theta$, and it is unclear how their results work for
the continuous case.

\texkdd{In a parallel work that will appear in the same proceeding, }
\texarxiv{In a parallel work, }
	He and Kempe study the
	same subject of robust influence maximization~\cite{HeKempe16}, but they follow
	the bi-criteria approximation approach of~\cite{krause2008robust}, and thus
	in general their results are orthogonal to ours.
In particular, they use essentially the same objective function, but they work on
	a finite set of influence spread functions $\Sigma$, and require to find
	$k\cdot \ln |\Sigma|$ seeds to achieve $1-1/e$ approximation ratio comparing to
	the optimal seed set of size $k$; when working on continuous parameter space
	$\Theta$, they show that it is equivalent to a finite spread function space
	of size $2^n$ and thus requiring $\varTheta(kn)$ seeds for a bi-criteria solution,
	which renders the bi-criteria solution useless.
Thus their bi-criteria approach is suitable when the set of possible spread functions
	$\Sigma$ is small.


Adaptive sampling for improving RIM bears some resemblance to pure exploration
bandit research~\cite{bubeck2011pure}, especially to combinatorial pure exploration
\cite{chen2014combinatorial} recently studied.
Both use adaptive sampling and achieve some optimization objective in the end.
However, the optimization problem modeled in combinatorial pure exploration
\cite{chen2014combinatorial} does not have a robustness objective.
Studying robust optimization together with combinatorial pure exploration
could be a potentially interesting topic for future research.
Another recent work \cite{lei2015online} uses online algorithms to maximize the 
	expected coverage of the union of influenced nodes in multiple rounds based on online
	feedbacks, and thus is different from our adaptive sampling objective: we use feedbacks
	to adjust adaptive sampling in order to find a seed set nearly maximizing
	the robust ratio after the sampling is done.

\section{Model and Problem Definition} \label{sect:model}

As in \cite{kempe2003maximizing}, the {\em independent cascade (IC)} model can be
equivalently modeled as a stochastic diffusion process from
seed nodes or as reachability from seed nodes
in random live-edge graphs.
For brevity, we provide the live-edge graph description below.
Consider a graph $G=(V,E)$ comprising a set $V$ of nodes and a set $E$ of directed edges,
where every edge $e$ is associated with probability $p_e \in [0,1]$, 
and let $n = |V|$ and $m = |E|$.
To generate a random live-edge graph, we declare each edge $e$
as {\em live}
if flipping a biased random coin with probability $p_e$ returns success,
declare $e$ as {\em blocked} otherwise (with probability $1-p_e$).
The randomness on all edges are mutually independent.
We define the subgraph $L$ consisting of $V$ and the set of
live edges as the (random) {\em live-edge graph}.
Given any set $S \subseteq V$ (referred as {\em seeds}), let $R_L(S) \subseteq V$ denote the {\em reachable set} of nodes from $S$ in live-edge graph $L$, i.e.,
(1) $S\subseteq R_L(S)$, and (2) for a node $v\notin S$, $v \in R_L(S)$ iff there is a path in $L$ directing from some node in $S$ to $v$.


For convenience, we use {\em parameter vector} $\theta=(p_e)_{e\in E}$ to denote the probabilities on all edges.
The {\em influence spread} function $\sigma_{\theta}(S)$ is defined as the expected size of the reachable set from $S$, that is
\[
\sigma_{\theta}(S) := \sum_{L}\Pr_{\theta}[L]\cdot |R_L(S)| \mbox{,}
\]
where $\Pr_{\theta}[L]$ is the probability of yielding live-edge graph $L$ under vector $\theta$.
From \cite{kempe2003maximizing}, we know that the influence spread function is non-negative ($\forall S \subseteq V$, $\sigma_{\theta}(S) \geq 0$), monotone ($\forall S \subseteq T \subseteq V$, $\sigma_{\theta}(S) \leq \sigma_{\theta}(T)$), and
submodular ($\forall S \subseteq T \subseteq V$, $\forall v \in V$ $\sigma_{\theta}(S \cup \{ v \}) - \sigma_{\theta}(S) \geq \sigma_{\theta}(T \cup \{ v \}) - \sigma_{\theta}(T)$).

The well-known problem  of \emph{Influence Maximization}
raised in \cite{kempe2003maximizing} is stated in the following.
\begin{problem}[Influence Maximization \cite{kempe2003maximizing}]
	Given a graph $G=(V,E)$, parameter vector $\theta=(p_e)_{e\in E}$ and a fixed budget $k$,
	we are required to find a seed set $S \subseteq V$ of $k$ vertices, such that the influence spread function $\sigma_{\theta}(S)$ is maximized, that is,
	\begin{align*}
	S^*_{\theta} := \argmax_{S \subseteq V, |S| = k} \sigma_{\theta}(S) \mbox{.}
	\end{align*}
\end{problem}
It has been shown that Influence Maximization problem is NP-hard \cite{kempe2003maximizing}.
Since the objective function $\sigma_{\theta}(S)$ is submodular,
we have a $(1-\frac{1}{e})$ approximation using standard greedy policy $\greedy(G,k,\theta)$ in Algorithm~\ref{alg:greedy} (assuming a value oracle on function $\sigma_\theta(\cdot)$).
Let $S^{g}_{\theta}$ be the solution of $\greedy(G,k,\theta)$.
As a convention, we assume that both optimal seed set $S^*_{\theta}$ and greedy seed set $S^{g}_{\theta}$
in this paper are of fixed size $k$ implicitly.

On the other hand, it is proved by Feige~\cite{feige1998threshold} that such an approximation ratio
could not be improved for $k$-max cover problem, which is a special case of the
influence maximization problem under the IC model.

\begin{algorithm}[t]
	\begin{algorithmic}[1]
		\REQUIRE Graph $G$, budget $k$, parameter vector $\theta$
		\STATE $S_0 \gets \emptyset$
		\FOR {$i = 1,2,\ldots,k$}
		\STATE $v \gets \argmax_{v \notin S_i} \left\{ \sigma_{\theta}(S_{i-1}\cup \{v\})-\sigma_{\theta}(S_{i-1}) \right\}$
		\STATE $S_i \gets S_{i-1} \cup \{v\}$
		\ENDFOR
		\RETURN $S_k$
	\end{algorithmic}
	\caption{{\sf Greedy}($G,k,\theta$)}
	\label{alg:greedy}
\end{algorithm}

However, the knowledge of the probability on edges is usually acquired by learning from the real-world
data \cite{saito2008prediction,tang2009social,goyal2010learning,gomez2011uncover,Netrapalli12},
and the obtained estimates always have some inaccuracy comparing to the
true value.
Therefore, it is natural to assume that, from observations of edge $e$,
we can obtain the statistically significant neighborhood $[l_e,r_e]$,
i.e., the {\em confidence interval} where the true probability $p_e$
lies in with high probability.
This confidence interval prescribes the uncertainty on the true probability
$p_e$ of the edge $e$, and such uncertainty on edges may adversely
impact the influence maximization task.
Motivated by this, we study the problem of {\em robust influence maximization}
as specified below.

%

Suppose for every edge $e$, we are given an interval $[l_e,r_e]$ ($0\le l_e\le r_e\le 1$)
indicating the range of the probability, and the ground-truth probability $p_e \in [l_e,r_e]$ of this edge is unknown.
Denote $\Theta=\times_{e\in E}[l_e,r_e]$ as the \emph{parameter space} of network $G$, and
$\theta=(p_e)_{e \in E}$ as the latent parameter vector.
Specifically, let $\theta^{-}(\Theta)=(l_e)_{e\in E}$ and $\theta^{+}(\Theta)=(r_e)_{e\in E}$ as the minimum and maximum parameter vectors, respectively, and when the context is clear, we would only use $\theta^{-}$ and $\theta^+$.
For a seed set $S \subseteq V$ and $|S| = k$, define its \emph{robust ratio} under
parameter space $\Theta$ as
\begin{equation}\label{robustratio}
g(\Theta,S) := \min_{\theta \in \Theta} \frac{\sigma_{\theta}(S)}{\sigma_{\theta}(S^{*}_{\theta})} \mbox{,}
\end{equation}
where $S^{*}_{\theta}$ is the optimal solution of size $k$ when the probability on every edge is given by $\theta$.

Given $\Theta$ and solution $S$, the robust ratio $g(\Theta,S)$
characterizes the {\em worst-case}
ratio of influence spread of $S$ and the underlying optimal one,
when the true probability vector $\theta$ is unknown (except knowing that
$\theta \in \Theta$).
Then, the {\em Robust Influence Maximization} (RIM) problem is defined as follows.
\begin{problem}[Robust Influence Maximization]
	Given a graph $G=(V,E)$, parameter space $\Theta=\times_{e\in E}[l_e,r_e]$ and a fixed budget $k$, we are required to find a set $S\subseteq V$ of $k$ vertices,
	such that robust ratio $g(\Theta,S)$ is maximized, i.e.,
	\begin{align*}
	S^*_{\Theta}
	:= \argmax_{S \subseteq V, |S| = k} g(\Theta,S)
	=  \argmax_{S \subseteq V, |S| = k} \min_{\theta\in \Theta}\frac{\sigma_{\theta}(S)}{\sigma_{\theta}(S^{*}_{\theta})} \mbox{.}
	\end{align*}
\end{problem}

%
The objective of this problem is to find a seed set $S^*_{\Theta}$ that has the largest robust ratio, that is, $S^*_{\Theta}$ should maximize the
worst-case ratio between its influence spread and the optimal influence spread,
when the true probability vector $\theta$ is unknown.
When there is no uncertainty, which means $\Theta$ collapses to the true
probability $\theta$, we can see that the RIM problem is reduced back
to the original influence maximization problem.

In RIM, the knowledge of the confidence interval is assumed to be the input.
Another interpretation is that, it can be viewed as given an estimate of probability vector
$\hat{\theta} = (\hat{p}_e)_{e \in E}$ with a perturbation level $\delta_e$ on each edge $e$,
such that the true probability $p_e \in [\hat{p}_e - \delta_e, \hat{p}_e + \delta_e] = [l_e, r_e]$,
which constitutes parameter space $\Theta = \times_{e \in E}[l_e,r_e]$.
Notice that, in reality, this probability could be obtained via edge samplings,
i.e., we make samples on edges and compute the fraction of times that the edge is live.
On the other hand, we can also observe information cascades on each edge when collecting
the trace of diffusion in the real world,
so that the corresponding probability can be learned.

However, when the amount of observed information cascade is small, the best robust ratio $\max_{S}g(\Theta,S)$ for the given $\Theta$ can be low so that the output for a RIM algorithm does not have a good enough guarantee of the performance in the worst case. Then a natural question is, given $\Theta$, how to further make samples on edges (e.g., activating source node $u$ of an edge $(u,v)$ and see if the sink node $v$ is activated through edge $e$) so that $\max_{S}g(\Theta,S)$ can be efficiently improved? To be specific, how to make samples on edges and output $\Theta'$ and $S'$ according to the outcome so that (a) with high probability the true value $\theta$ lies in the output parameter space $\Theta'$, where the randomness is taken according to $\theta$,
and (b) $g(\Theta',S')$ is large.
This sub-problem is called \emph{Sampling for Improving Robust Influence Maximization}, and will be addressed in Section~\ref{sect:sample}.

\section{Algorithm and Analysis for RIM} \label{sect:rim}


Consider the problem of RIM, parameter space $\Theta = \times_{e \in E}[l_e,r_e]$ is given, and we do not know the true probability $\theta \in \Theta$.
Let $\theta^{-}=(l_e)_{e\in E}$ and $\theta^{+}=(r_e)_{e\in E}$.

Our first observation is that, when $\Theta$ is a single vector ($l_e = r_e$, $\forall e \in E$),
it is trivially reduced to the classical Influence Maximization problem.
Therefore, we still have the following hardness result on RIM \cite{kempe2003maximizing,feige1998threshold}:

\begin{theorem} \label{pro:RIM-nphard}
	RIM problem is \NP-hard, and for any $\varepsilon > 0$, it is \NP-hard
	to find a seed set $S$ with robust ratio at least
	$1-\frac{1}{e} + \varepsilon$.
\end{theorem}

To circumvent the above hardness result, we develop algorithms that achieves reasonably
large robust ratio.
When we are not allowed to make new samples on the edges to improve the input interval, it is natural to
utilize the greedy algorithm of submodular maximization in \cite{kempe2003maximizing} (i.e., Algorithm~\ref{alg:greedy})
as the subroutine to calculate the solution.
In light of this, we first propose Lower-Upper Greedy Algorithm and the solution-dependent bound for $g(\Theta,S)$,
and then discuss $g(\Theta,S)$ in the worst-case scenario.

\subsection{Lower-Upper Greedy Algorithm}

\begin{algorithm}[t]
	\begin{algorithmic}[1]
		\REQUIRE Graph $G = (V,E)$, budget $k$, parameter space $\Theta = \times_{e \in E}[l_e, r_e]$
		\STATE $S_{\theta^{-}}^g \gets \greedy(G,k,\theta^-)$ 
		\STATE $S_{\theta^{+}}^g \gets \greedy(G,k,\theta^+)$
		\RETURN $\argmax_{S \in \left\{ S_{\theta^{-}}^g, S_{\theta^{+}}^g \right\}} \left\{ \sigma_{\theta^-}(S) \right\}$
	\end{algorithmic}
	\caption{{\lugreedy}($G,k,\Theta$)}
	\label{alg:lugreedy}
\end{algorithm}

Given parameter space $\Theta=\times_{e\in E}[l_e,r_e]$ with the minimum and maximum parameter vectors $\theta^{-}=(l_e)_{e\in E}$ and $\theta^{+}=(r_e)_{e\in E}$,
our \emph{Lower-Upper Greedy algorithm} ($\lugreedy(G, k, \Theta)$) is described in Algorithm~\ref{alg:lugreedy}
which outputs the best seed set $S^\lu_\Theta$ for the minimum parameter vector $\theta^-$ such that
\begin{align}\label{def:lu-greedy-solution}
S^\lu_\Theta := \argmax_{S \in \left\{ S_{\theta^{-}}^g, S_{\theta^{+}}^g \right\}} 
\left\{ \sigma_{\theta^-}(S) \right\}.
\end{align}

To evaluate the performance of this output, we first define the \emph{gap ratio} $\alpha(\Theta) \in [0,1]$ of the input parameter space to be
\begin{equation} \label{eq:def-alpha}
\alpha(\Theta):=
\frac{\sigma_{\theta^{-}}(S^\lu_\Theta)}
{\sigma_{\theta^{+}}(S_{\theta^{+}}^g)}  \mbox{.}
\end{equation}
Then, {\lugreedy} achieves the following result:

\begin{theorem}[solution-dependent bound] \label{thm:main}
	Given a graph $G$, parameter space $\Theta$ and budget limit $k$, {\lugreedy} outputs a seed set
	$S^\lu_\Theta$ of size $k$ such that
	\begin{displaymath}
	g(\Theta, S^\lu_\Theta) \ge \alpha(\Theta)\left(1-\frac{1}{e}\right) \mbox{,}
	\end{displaymath}
	where $\alpha(\Theta):=
	\frac{\sigma_{\theta^{-}}(S^\lu_\Theta)}
	{\sigma_{\theta^{+}}(S_{\theta^{+}}^g)}$.
\end{theorem}

\begin{proof}
	For any seed set $S$,
	$g(\Theta,S) = \min_{\theta\in \Theta} \frac{\sigma_{\theta}(S)}{\sigma_{\theta}(S^{*}_{\theta})}$ 
	by definition.
	Obviously, it is a fact that $\sigma_\theta(S)$ is monotone on $\theta$ for any fixed $S$.
	From the definition of optimal solutions and the greedy algorithm, we can get
	$
	\sigma_{\theta}(S^{*}_{\theta}) \le \sigma_{\theta^+}(S^{*}_{\theta}) \le \sigma_{\theta^{+}}(S^{*}_{\theta^{+}})  \le \frac{ \sigma_{\theta^{+}}(S^{g}_{\theta^{+}}) }{1 - 1/e }.
	$
	Moreover, it can be implied that
	\begin{align*}
	g(\Theta,S) \ge \min_{\theta\in \Theta} \frac{ \sigma_{\theta}(S) }{ \sigma_{\theta^{+}}(S^{g}_{\theta^{+}}) } \left(1-  \frac{1}{e} \right)
	= \frac{ \sigma_{\theta^{-}}(S) }{ \sigma_{\theta^{+}}(S^{g}_{\theta^{+}}) } \left(1-  \frac{1}{e} \right).
	\end{align*}
	Use seed set $S^{\lu}_{\Theta}$ from {\lugreedy}, and it follows immediately that
	$
	g(\Theta, S^{\lu}_{\Theta}) 
	\ge \frac{ \sigma_{\theta^{-}}(S^{\lu}_{\Theta}) }{ \sigma_{\theta^{+}}(S^{g}_{\theta^{+}}) } \left(1-  \frac{1}{e} \right)
	= \alpha(\Theta) \left(1 - \frac{1}{e} \right).
	$
\end{proof}

%
%
%
%
%

We refer $\alpha(\Theta) (1- \frac{1}{e})$ as the {\em solution-dependent bound} of $g(\Theta,S^\lu_{\Theta})$  
that {\lugreedy} 
achieves, because it depends on the solution $S^\lu_{\Theta}$.
The good thing is that it can be evaluated once we have the solution, and then
we know the robust ratio must be at least this lower bound.
Note that the bound is good if $\alpha(\Theta)$ is not too small, and thus it in turn indicates that
the influence spread $\sigma_{\theta}(S^\lu_{\Theta})$
we find has a good performance under any probability vector $\theta \in \Theta$.

It is worth remarking that the choice of using 
$\alpha(\Theta) = \sigma_{\theta^{-}}(S^\lu_\Theta) / \sigma_{\theta^{+}}(S_{\theta^{+}}^g)$ 
as a measurement is for the following reasons:
(a) Intuitively, $S_{\theta^{-}}^g$ is expected to be the best possible seed set we can find that maximizes $\sigma_{\theta^{-}}(\cdot)$;
(b) Meanwhile, we consider $S_{\theta^{+}}^g$ as a potential seed set for
the later theoretical analysis (in the proof of \Cref{thm:uniform}), 
which requires the alignment of the same seed set for the numerator and denominator.
Thus, $\alpha(\Theta) \geq \max\{ \sigma_{\theta^{-}}(S_{\theta^{-}}^g), \sigma_{\theta^{-}}(S_{\theta^{+}}^g) \} / \sigma_{\theta^{+}}(S_{\theta^{+}}^g)$.
In particular, when $\theta^+$ and $\theta^-$ tend to the same value $\theta$, 
RIM is tending towards the classical Influence Maximization,
and thus the influence spread $\sigma_{\theta}(S^\lu_\Theta)$
can be close to the best possible result $\sigma_{\theta}(S_{\theta}^g)$. 
The approach adopted by {\lugreedy} is similar to the sandwich approximation used 
in~\cite{lu2015competition}.


The following example shows that for certain problem instances,
the gap ratio $\alpha(\Theta)$ of {\lugreedy} could match the robust ratio
$g(\Theta,S^\lu_\Theta)$, which also matches
the best possible robust ratio $\max_{|S|=k}g(\Theta,S)$.
\begin{example} \label{exp:tight}
	Consider a graph $G=(V,E)$ where the set of nodes are equally partitioned into $2k$ subsets $V=\cup_{i=1}^{2k}V_i$ such that every $V_i$ contains $t+1$ nodes.
	Let $V_i=\{v_i^{j}\mid 1\le j\le t+1\}$ and set $E=\cup_{i=1}^{2k}E_i$
	where $E_i=\{(v_i^{1},v_i^{j})\mid 2\le j\le t+1\}$.
	That is, every $(V_i,E_i)$ forms a star with $v_i^1$ being the node at the center,
	all stars are disconnected from one another.
	For the parameter space we set the interval on every edge to be $[l,r]$.
	When {\lugreedy} select $k$ nodes, since all $v^1_i$'s have the same (marginal)
	influence spread,  w.l.o.g., suppose that {\lugreedy} selects
	$\{v^1_1, v^1_2,\ldots, v^1_k \}$.
	Then if we set the true probability vector  $\theta\in \Theta$ such that $p_e=l$ for every $e \in \cup_{i=1}^{k} E_i$, and $p_e=r$
	for every $e\in \cup_{i=k+1}^{2k} E_i$, it is easy to check that
	$
		\max_{|S|=k}g(\Theta,S)=g(\Theta,S^\lu_{\Theta})=\alpha(\Theta)=\frac{1+tl}{1+tr}.
	$
\end{example}
The intuition from the above example is that, when there are many alternative choices
for the best seed set, and these alternative seed sets do not have much overlap
in their influence coverage, the gap ratio $\alpha(\Theta)$ is a good indicator
of the best possible robust ratio one can achieve.

%

In the next subsection, we will show that the best robust ratio could be very bad
for the worst possible graph $G$ and parameter space $\Theta$,
which motivates us
to do further sampling to improve $\Theta$.

\subsection{Discussion on the robust ratio} \label{sect:analysis-sub:hardness}


For the theoretical perspective, we show in this part that if we make no assumption or only add loose constraints to the input parameter space $\Theta$, then no algorithm will guarantee good performance for some worst possible graph $G$.
%

\begin{theorem}
	\label{thm:3-ratios}
	For RIM,
	\begin{enumerate}
		\item \label{pro:3-ratios-1}
		There exists a graph $G=(V,E)$ and parameter space $\Theta = \times_{e\in E} [l_{e}, r_{e}]$, such that
		\[
		\max_{|S|=k}g(\Theta,S)=\max_{|S|=k}\min_{\theta\in \Theta} \frac{\sigma_{\theta}(S)}{\sigma_{\theta}(S^{*}_{\theta})}=O\left(\frac{k}{n}\right) \mbox{.}
		\]
		
		\item \label{pro:3-ratios-2}
		There exists a graph $G=(V,E)$, constant $\delta = \varTheta\left(\frac{1}{n}\right)$ and
		parameter space $\Theta = \times_{e \in E}[l_{e}, r_{e}]$ where
		$r_{e} - l_{e} \leq \delta$ for every $e \in E$,
		such that
		\[
		\max_{|S|=k}g(\Theta,S)=O\left(\frac{\log n}{n}\right) \mbox{.}
		\]
		
		\item \label{pro:3-ratios-3}
		Consider random seeds set $\tilde{S}$ of size $k$.
		There exists a graph $G=(V,E)$, constant $\delta = \varTheta\left(\frac{1}{\sqrt{n}}\right)$ and
		parameter space $\Theta = \times_{e \in E}[l_{e}, r_{e}]$ where
		$r_{e} - l_{e} \leq \delta$ for every $e \in E$,
		we have
		\[
		\max_{\Omega} \min_{\theta\in\Theta}
		\E_{\tilde{S}\in \Omega}\left[
		\frac{\sigma_{\theta}(\tilde{S})}{\sigma_{\theta}(S_{\theta}^*)} \right]=O\left(\frac{\log n}{\sqrt{n}}\right) \mbox{,}
		\]
		where $\Omega$ is any probability distribution over seed sets of size $k$, and
		$\E_{\tilde{S}\in \Omega}[\cdot]$ is the expectation of random set
		$\tilde{S}$ taken from the distribution $\Omega$.
		
	\end{enumerate}
\end{theorem}


In the first case, we allow the input $\Theta$ to be an arbitrary parameter space.
It is possible that $\Theta=\times_{e\in E}[0,1]$ for some graph $G$, which means
there is no knowledge at all for edge probabilities.
Then any seed set may achieve $O\left(\frac{k}{n}\right)$-approximation
of the optimal solution in the worst case.
Intuitively, a selected seed set $S$ may rarely activate other nodes (i.e., $O(k)$), while optimal solution (to the latent $\theta$) may cover almost the whole graph (i.e., $\Omega(n)$).

In the second case, an additional constraint is assumed on the parameter space
$\norm{\theta^+ - \theta^-}_{\infty} \leq \delta$, i.e.,
for every $e\in E$, $r_e-l_e\le \delta$, to see if we could obtain a better performance when $\delta$ is small.
However, even though $\delta$ is in the order of $O(1/{n})$, the robust ratio can be as small as $O(\log{n}/{n})$.
The proof is related to the phase transition in the Erd\H{o}s-R\'{e}nyi graph
for the emergence of giant component.
In particular, if we have a graph $G$ consisting of two disconnected, equal-sized
Erd\H{o}s-R\'{e}nyi random graphs with edge probabilities close to the
critical value of generating a giant connected component, then whenever
we select a seed in one component, that component could be just below
the threshold resulting in $O(\log n)$ influence spread while the other
component is just above the threshold leading to $\varTheta(n)$ influence spread.
Thus, the worst-case ratio for any one-node seed set is always
$O(\log{n}/{n})$.
A similar discussion can be found in \cite{he2015stability}.

In the third case, we allow the algorithm to be randomized,
namely the output seed set $\tilde{S}$ is a random set of size $k$.
Even in this case, the robust ratio could be as bad as
$O(\log n/\sqrt{n})$.

\section{Sampling for Improving RIM}\label{sect:sample}

From the previous section, we propose {\lugreedy} algorithm to check the solution-dependent bound of the robust ratio,
and point out the worse-case bound could be small if $\Theta$ is not assumed to be tight enough.

Theorem~\ref{thm:3-ratios} in the previous subsection
points out that the best possible robust ratio $\max_{S}g(\Theta,S)$ can be
too low so that the output for RIM could not provide us with a satisfying seed set in the worst case.
Then a natural question is: given the input $\Theta$, can we make efficient samples
on edges
so that $\Theta$ is narrowed into $\Theta'$ (this means the true $\theta\in \Theta'$ with high probability) and then output a seed set $S'$ that makes $g(\Theta',S')$ large?
This problem is called \emph{Sampling for Improving RIM}.

In this section we study both uniform sampling and adaptive sampling
for improving RIM.
According to the Chernoff's bound, the more samples we make on an edge, the narrower
the confidence interval we get that guarantees the true probability
to be located within the confidence interval with a desired probability of
confidence.
After sampling to get a narrower parameter space, we could use 
{\lugreedy} algorithm 	
to get the seed set.

\subsection{Uniform Sampling} \label{sect:analysis-sub:Uniform}



In Sampling for improving RIM, the goal is to design a sampling and maximization algorithm $\A$ that outputs $\Theta'$ and $S'$ such that with high probability the robust ratio of $S'$ in $\Theta'$ is large.
After sampling edges, we can use Chernoff's bound to compute the confidence interval,
and the confidence interval can be further narrowed down with more samples.
However, the key issue is to connect the width of confidence interval with
the stability of influence spread.
We propose two ideas
exploiting properties of additive and multiplicative confidence interval respectively
to this issue, and incorporate into Uniform Sampling algorithm (in Algorithm~\ref{alg:uniform-sampling})
with theoretical justification (in Theorem~\ref{thm:uniform}).




Our first idea is inspired by the following lemma from \cite{ChenWY14a} to build the connection in the additive form.

\begin{lemma}[Lemma~7 in \cite{ChenWY14a}]
	\label{lem:add}
	Given graph $G$ and parameter space $\Theta$ such that $\forall \theta_1,\theta_2\in \Theta$,  $\norm{\theta_1-\theta_2}_{\infty}\leq \delta$, then, $\forall S\subseteq V$,
	\[
	\abs{\sigma_{\theta_1}(S)-\sigma_{\theta_2}(S)} \leq mn\delta \mbox{.}
	\]
\end{lemma}

We use a tight example (in the order of $|V|$ and $|E|$) to illustrate the connection and give an insight of this lemma as follows.
Consider graph $G = (V,E)$ with $|V|=n$ and $|E| = m$ ($m \gg n$).
Let $G$ be two disjoint cycles, each containing exactly $\frac{n}{2}$ nodes and $\frac{n}{2}$ edges.
We arbitrarily assign the rest $m-n$ edges between two cycles.
Then, for every edge $e$ in the cycle, the interval is $l_e = r_e = 1$,
and $l_e = 0$, $r_e = \delta$ for those between two cycles,
which constitutes $\Theta = \times_{e \in E} [l_e, r_e]$.
Suppose $\delta > 0$ is sufficiently small, and let budget $k=1$. For any single-node set $S$,
it is easy to check that for $\theta^- = (l_e)_{e \in E}$, $\sigma_{\theta^-}(S) = \frac{n}{2}$,
and for $\theta^+ = (r_e)_{e \in E}$,
$\sigma_{\theta^+}(S) \approx \frac{n}{2} + \frac{n}{2} (m - n) \delta$,
thus $\abs{\sigma_{\theta^+}(S)-\sigma_{\theta^-}(S)} \approx \frac{1}{2} n(m-n)\delta$ in this case.
As a comparison, from Lemma~\ref{lem:add}, we know that $\abs{\sigma_{\theta^+}(S)-\sigma_{\theta^-}(S)} \leq m n \delta$.

Therefore, the above lemma establishes the guidance that we may sample every edge for sufficient times to shrink their confidence intervals in $\Theta$,
and feed {\lugreedy} with $\Theta$ as same as solving RIM, then the performance is guaranteed by Theorem~\ref{thm:main},
which matches our intuition that {\lugreedy} performs well with the satisfactory $\Theta$.

On the other hand, our second idea is to use the multiplicative confidence interval to reduce the fluctuation of influence spread,
then {\lugreedy} still applies.
The next lemma is crucial to achieve this goal.

\begin{lemma}
	\label{lem:mul}
	Given graph $G=(V,E)$ and parameter space $\Theta$. If there exists $\lambda \geq 0$, for all edge $e\in E$, s.t.,
	$r_e \leq (1+\lambda)l_e$, then for any nonempty set $S \subseteq V$,
	\begin{align}
	\frac{\sigma_{\theta^+}(S)}{\sigma_{\theta^-}(S)} \leq (1+\lambda)^{n} \mbox{,}
	\end{align}
	and
	\begin{equation}
	\max_{|S|= k}\min_{\theta\in\Theta}\frac{\sigma_{\theta}(S)}{\sigma_{\theta}(S_{\theta}^{*})}\geq (1+\lambda)^{-n} \mbox{.}
	\end{equation}
\end{lemma}
In this lemma, the ratio of influence spread can be bounded based on the relation of $l_e$ and $r_e$ in the multiplicative form.


%
%

To unify both ideas mentioned above, we propose \emph{Uniform Sampling for RIM} algorithm ({\sf US-RIM}) in Algorithm~\ref{alg:uniform-sampling},
and the theoretical result is presented in Theorem~\ref{thm:uniform}.
Basically, the algorithm samples every edge with the same number of
times, 
and use {\lugreedy} to obtain the seed set.
We set different $t$ and $\delta_e$ for the two ideas.
Henceforth,
we explicitly refer the first setting as {\em Uniform Sampling with Additive form} ({\sf US-RIM-A}), and the second one as {\em Uniform Sampling with Multiplicative form} ({\sf US-RIM-M}).

\begin{algorithm}[t]
	\begin{algorithmic}[1]
		\REQUIRE Graph $G=(V,E)$, budget $k$, $(\epsilon, \gamma)$
		\ENSURE Parameter space $\Theta_{out}$, seed set $S_{out}$
		\FORALL{$e \in E$}
		\STATE Sample $e$ for $t$ times, and observe $x^1_e, \ldots, x^t_e$
		\STATE $p_e\gets\frac{1}{t}\sum_{i=1}^{t}x_e^i$, and set $\delta_e$ according to Theorem~\ref{thm:uniform}
		\STATE $r_e\gets\min\{1,p_e+\delta_e\}$, $l_e\gets\max\{0,p_e-\delta_e\}$
		\ENDFOR
		\STATE 
		$\Theta_\text{out} \gets \times_{e \in E} [l_e,r_e]$
		\STATE $S_\text{out} \gets \lugreedy(G,k,\Theta_\text{out})$ 
		\RETURN $(\Theta_\text{out}$,$S_\text{out})$
	\end{algorithmic}
	\caption{{\sf US-RIM}}
	\label{alg:uniform-sampling}
\end{algorithm}


\begin{theorem}\label{thm:uniform}
	Given a graph $G=(V,E)$, budget $k$, and accuracy parameter $\epsilon,\gamma>0$, let $n=|V|$ and $m=|E|$, then for any unknown ground-truth parameter vector $\theta=(p_e)_{e \in E}$, Algorithm {\sf US-RIM} outputs $(\Theta_\text{out}$,$S_\text{out})$ such that
	\[
	g(\Theta_\text{out}, S_\text{out})\ge \left(1-\frac{1}{e}\right)(1-\epsilon),
	\]
	with $\Pr[\theta \in \Theta_\text{out}]\ge 1-\gamma$,
	where the randomness is taken according to $\theta$,
	if we follow either of the two settings:
	\begin{enumerate}
		\item \label{thm-add-case}
		Set $t=\frac{2m^2n^2 \ln (2m/\gamma)}{k^2\epsilon^2}$, and for all $e$, set $\delta_e=\frac{k\epsilon}{mn}$;
		\item \label{thm-mul-case}
		Assume we have $p'$ such that $0 < p' \leq \min_{e\in E} p_e$,
		set $t=\frac{3 \ln (2m/\gamma)}{ p' } \left( \frac{2n}{\ln (1/1-\epsilon)} + 1 \right)^2$,
		and for all edge, set $\delta_e=\frac{1}{n} p_e\log\frac{1}{\gamma}$.
	\end{enumerate}
\end{theorem}

In general, the total number of samples summing up all edges is $O(\frac{m^3n^2\log (m/\gamma)}{k^2\epsilon^2})$ for {\USRIMA},
and $O(\frac{mn^2 \log (m/\gamma)}{p' \epsilon^2})$ for {\USRIMM} with an additional constant $p'$, the lower bound probability on all edge probabilities.
The difference is that the former has a higher order of $m$, and the latter requires the knowledge of $p'$ and has an extra dependency on $O(1/p')$.
Since the sample complexity for both settings can be calculated in advance,
one may compare the values and choose the smaller one when running the uniform sampling algorithm.
An intuitive interpretation is that: (1) with high probability ($\geq 1-\gamma$),
the algorithm always outputs an $(1-\frac{1}{e}-\epsilon)$-approximation solution guaranteed by {\USRIMA};
(2) if $p'=\Omega(\frac{k^2}{m^2})$ (it is a loose assumption naturally satisfied in practice),
we may choose {\USRIMM} to achieve better sample complexity.



\subsection{Non-uniform and Adaptive Sampling} \label{sect:heuristic}

In a real network, the importance of edges in an influence diffusion process varies significantly.
Some edges may have larger influence probability than others or connect two important nodes in the network.
Therefore, in sampling it is crucial to sample edges appropriately.
Moreover, we can adapt our sampling strategy dynamically to put more sampling effort on
critical edges when we learn the edge probabilities
more accurately over time.

For convenience, given graph $G = (V, E)$, we define \emph{observation set} $\mathcal{M} = \left\{ M_e \right\}_{e \in E}$ as a collection of sets, where
$M_e = \{ x^1_e, x^2_e, \cdots, x^{t_e}_e \}$ denotes observed values of edge $e$ via the first $t_e$ samples on edge $e$.
We allow that a parameter space $\Theta_0 \subseteq \times_{e \in E} [0,1]$ is given,
which can be obtained by some initial samples $\mathcal{M}_0$
(e.g., uniformly sample each edge of the graph for a fixed number of times).
%
%
%
%

The following lemma is used to calculate the confidence interval, which
is a combination of additive and multiplicative Chernoff's Bound.
We adopt this bound in the experiment since some edges in the graph have large influence probability while others have small ones,
but using either additive or multiplicative bound may not be good enough
to obtain a small confidence interval.
The following bound is adapted from \cite{badanidiyuru2013bandit}
and is crucial for us in the experiment.


\begin{lemma} \label{lem:conf}
	For each $e \in E$, let $M_e = \left\{ x^1_e, x^2_e, \dots, x^{t_e}_e \right\}$ be samples of $e$ in $\mathcal M = \{ M_e \}_{e \in E}$, and $t_e$ be the sample number.
	Given any $\gamma > 0$, let confidence intervals for all edges be
	$\Theta = \times_{e\in E} [l_e, r_e]$, such that, for any $e \in E$,
	\begin{equation*}
	\begin{aligned}
	l_e &= \min\left\{\hat{p}_e + \frac{c_e^2}{2} - c_e\sqrt{\frac{c_e^2}{4} + \hat{p}_e}, ~0\right\}\\
	r_e &= \max\left\{\hat{p}_e + \frac{c_e^2}{2} + c_e\sqrt{\frac{c_e^2}{4} + \hat{p}_e}, ~1\right\},
	\end{aligned}
	\end{equation*}
	where $\hat{p}_e=\frac{ \sum_{i=1}^{t_e} x^{i}_e }{t_e}$, $c_e = \sqrt{\frac{3}{t_e} \ln\frac{2m}{\gamma}}$.
	Then, with probability at least $1-\gamma$, the true probability $\theta= \left(p_e\right)_{e\in E}$ satisfies that $\theta\in\Theta$.
\end{lemma}





Our intuition for non-uniform sampling is that the edges along the information cascade
of important seeds determine the influence spread, and henceforth they should be estimated more accurately than other edges not along important information
cascade paths.
Thus, we use the following \emph{Information Cascade Sampling} method to select edges.
Starting from the seed set $S$, once node $v$ is activated, $v$ will try to activate its out-neighbors.
In other words, 
for every out-edge $e$ of $v$, denote $t_e$ as the number of samples,
then $e$ will be sampled once to generate a new observation $x_e^{t_e}$ based on the latent Bernoulli distribution with success probability $p_e$,
and $t_e$ will be increased by $1$.
The process goes on until the end of the information cascade.

We propose \emph{Information Cascade Sampling for RIM} ({\ICSRIM}) algorithm in Algorithm~\ref{alg:information-cascade-sampling},
which adopts information cascade sampling described above to select edges.

\begin{algorithm}[t]
	\begin{algorithmic}[1]
		\REQUIRE Graph $G=(V,E)$, budget $k$, initial sample $\mathcal{M}_0$, threshold $\kappa$, $\gamma$.
		\ENSURE Parameter space $\Theta_\text{out}$, seed set $S_\text{out}$
		\STATE $i\gets 0$
		\REPEAT
		\STATE Get $\Theta_i$ based on $\mathcal{M}_i$ (see Lemma~\ref{lem:conf}).
		\STATE $S^\lu_{\Theta_i} = \lugreedy(G,k, \Theta_i)$ 
		\STATE $\mathcal{M}_{i+1} \gets \mathcal{M}_i$
		\FOR{$j=1,2,\ldots,\tau$}
		\STATE Do information cascade with the seed set $S^\lu_{\Theta_i}$ 
		\STATE During the cascade, once $v\in V$ is activated, sample all out-edges of $v$ and update $\mathcal{M}_{i+1}$
		\ENDFOR
		\STATE $i\gets i+1$
		\UNTIL{$\alpha(\Theta_i)>\kappa$}
		\STATE $S_\text{out} \gets S^\lu_{\Theta_{i-1}}$ 
		\STATE $\Theta_\text{out} \gets \Theta_{i-1}$
		\RETURN $(\Theta_\text{out}, S_\text{out})$
	\end{algorithmic}
	\caption{{\ICSRIM}$(\tau)$: Information Cascade Sampling}
	\label{alg:information-cascade-sampling}
\end{algorithm}

Algorithm~\ref{alg:information-cascade-sampling} is an iterative procedure. In the $i$-th iteration,
Lemma~\ref{lem:conf} is used to compute the confidence interval $\Theta_i$ from observation set $\mathcal{M}_i$.
Then according to $\Theta_i$, we find the lower-upper greedy set $S^{\lu}_{\Theta_{i}}$ and use information cascade to
update observation set $\mathcal{M}_{i+1}$ by absorbing new samples.

Since the robust ratio $g(\Theta, S^{\lu}_{\Theta_{i}})$ cannot be calculated efficiently, we will calculate $\alpha(\Theta)$ (defined in \eqref{eq:def-alpha}) instead.
In our algorithm, we use a pre-determined threshold $\kappa$ ($\kappa \in (0,1)$) as the stopping criteria.
Therefore, for $S_\text{out}$, the robust ratio $g(\Theta,S_\text{out})\ge \alpha(\Theta)\left(1-\frac{1}{e}\right) > \kappa\left(1-\frac{1}{e}\right)$ is guaranteed by Theorem~\ref{thm:main},
and the true probability $\theta \in \Theta_\text{out}$ holds with probability at least $1-\gamma$ due to Lemma~\ref{lem:conf}. 

%

Compared with information cascade sampling method, calculating a greedy set is time-consuming.
Therefore in Algorithm \ref{alg:information-cascade-sampling}, we call $\lugreedy$ once every $\tau$ rounds of information cascades
to reduce the cost.

\section{Empirical Evaluation} \label{sect:experiments}

We conduct experiments on two datasets,
Flixster\footnote{http://www.cs.sfu.ca/$\sim$sja25/personal/datasets/} and 
NetHEPT\footnote{http://research.microsoft.com/en-us/people/weic/projects.aspx} 
to verify the robustness of influence maximization and our sampling methods.

%

\subsection{Experiment Setup}

\subsubsection{Data Description}

\paragraph*{Flixster}
The Flixster dataset is a network of American social movie discovery service (www.flixster.com). To transform the dataset into a weighted graph, each user is represented by a node, and a directed edge from node $u$ to $v$ is formed
if $v$ rates one movie shortly after $u$ does so on the common movie.
The dataset is analyzed in \cite{barbieri2013topic}, and the influence probability are learned by the topic-aware model.
We use the learning result of \cite{barbieri2013topic} in our experiment, which is a graph containing 29357 nodes and 212614 directed edges.
There are 10 probabilities on each edge, and each probability represents the influence from the source user to the sink on a specific topic.
Since most movies belong to at most two topics,
we only consider 3 out of 10 topics in our experiment, and get two induced graphs whose number of edges are 23252 and 64934 respectively. For the first graph, probabilities of topic 8 are directly used as the ground truth parameter (termed as Flixster(Topic~8)).
For the second graph, we mix the probabilities of Topic 1 and Topic 4 on each edge evenly to obtain the ground-truth probability (termed as as Flixster(Mixed)).
After removing isolated nodes, the number of nodes in the two graphs are 14473 and 7118 respectively.

In~\cite{barbieri2013topic}, the probability for every edge $(u,v)$ is learned
by rating cascades that reach $u$ and may or may not reach $v$, and in this
cases we view that edge $(u,v)$ are sampled.
According to the data reported in~\cite{barbieri2013topic}, on average
every edge is sampled $318$ times for their learning process.
We then use $318$ samples on each edge as our initial sample
${\cal M}_0$.

\paragraph*{NetHEPT}
The NetHEPT dataset \cite{chen2009efficient} is extensively used in many influence
maximization studies.
It is an academic collaboration network from the "High Energy
Physics-Theory" section of arXiv form 1991 to 2003, where nodes represent the authors
and each edge in the network represents one paper co-authored by two nodes.
It contains $15233$ nodes and $58891$ undirected edges (including duplicated edges).
We remove those duplicated edges and obtain a directed graph $G=(V,E), |V|=15233, |E|=62774$ (directed edges).
Since the NetHEPT dataset does not contain the data of influence probability on edges,
we set the probability on edges according to the \emph{weighted cascade} model \cite{kempe2003maximizing}
as the ground truth parameter, i.e.,
$\forall e = (v, u)\in E$, let $x_u$ be the in-degree of $u$ in the
edge-duplicated graph, $y_{e}$ be the number of edges connecting node $v$ and $u$,
then the true probability is $p_e = 1 - (1-\frac{1}{x_u})^{y_e}$.
%
Following the same baseline of Flixster, we initially sample each edge 
for 318 times as $\mathcal{M}_0$. 

\subsubsection{Algorithms}
\label{algorithms}

We test both the uniform sampling algorithm {\USRIM} and the adaptive sampling
algorithm {\ICSRIM}, as well as another adaptive algorithm
{\OESRIM} (Out-Edge Sampling) as the baseline (to be described shortly).
Each algorithm is given a graph $G$ and initial observation set $\mathcal{M}_0$.
Note that the method to estimate the parameter space based on sampling results in Algorithm~\ref{alg:uniform-sampling} and Algorithm~\ref{alg:information-cascade-sampling} are different. In order to make the comparison meaningful, in this section, for all three algorithms, a common method according to Lemma~\ref{lem:conf} is used to estimate the parameter space. In all tests, we set the size of the seed set $k=50$. To reduce the running time, we use a faster approximation algorithm PMIA (proposed in \cite{chen2010scalable}) to replace the well known greedy algorithm purposed in \cite{kempe2003maximizing} in the whole experiment. The accuracy requirement $\gamma=o(1)$ is set to be $\gamma=m^{-0.5}$ where $m$ is the number of edges.

\paragraph*{\USRIM}
The algorithm is slightly modified from Algorithm~\ref{alg:uniform-sampling} for a better  comparison of performance. The modified algorithm proceeds in an iterative fashion: In each iteration, the algorithm makes $\tau_1$ 
samples on each edge, updates $\Theta$ according to Lemma~\ref{lem:conf} and computes $\alpha(\Theta)$. The algorithm stops when $\alpha(\Theta)\ge \kappa=0.8$.
$\tau_1$ is set to 1000, 1000, 250 for NetHEPT, Flixster(Topic~8), Flixster(Mixed), respectively
to achieve fine granularity 
and generate visually difference of $\alpha(\Theta)$ in our results.

\paragraph*{\ICSRIM}
As stated in Algorithm~\ref{alg:information-cascade-sampling}, in each iteration, the algorithm do 
$\tau_2 = 5000$ 
times information cascade sampling based on the seed set from the last iteration,
and then it updates $\Theta$ according to Lemma~\ref{lem:conf}, computes $\alpha(\Theta)$ and uses {\lugreedy} algorithm to compute the seed set for the next round. The algorithm stops when $\alpha(\Theta)\ge \kappa=0.8$.

\paragraph*{\OESRIM}
This algorithm acts as a baseline, and it proceeds in the similar way to {\ICSRIM}.
Instead of sampling information cascades starting from the current seed set
as in {\ICSRIM}, {\OESRIM} only sample {\em out-edges} from the seed set.
More specifically, in each iteration, the algorithm samples $5000$ 
times of all out-edges of the seed set from last iteration, for the three graphs respectively, and then it updates $\Theta$ according to Lemma~\ref{lem:conf}, computes $\alpha(\Theta)$ and uses {\lugreedy} algorithm to compute the seed set for the next round.
Note that for {\OESRIM}, $\alpha(\Theta)$ remains small (with the increase of the number of samples) and cannot exceed the threshold $\kappa$ even the iteration has been processed for a large number of times,
therefore we will terminate it when $\alpha(\Theta)$ is stable.

\subsubsection{$\bar{\alpha}$ as a Upper Bound}
Theorem~\ref{thm:main} shows that $\alpha(\Theta)\left(1-\frac{1}{e}\right)$ is
a lower bound for the robust ratio $g(\Theta,S^\lu_{\Theta})$.
We would also like to find some upper bound of $g(\Theta, S^\lu_{\Theta})$:
If the upper bound is reasonably close to the lower bound or match in trend of
changes,  it indicates that $\alpha(\Theta)\left(1-\frac{1}{e}\right)$ is
a reasonable indicator of the robust ratio achieved by the {\lugreedy}
output $S^\lu_{\Theta}$.
For any $\theta\in \Theta$, we define $\bar{\alpha}(\Theta, \theta)
= \frac{\sigma_{\theta}\left( ^\lu_{\Theta} \right)}{\sigma_{\theta}(S_{\theta}^g)}$.
The following shows that $\bar{\alpha}(\Theta, \theta)$ is an upper bound for
$g(\Theta,S^\lu_{\Theta})$:
\begin{equation*}
\bar{\alpha}(\Theta, \theta) =
\frac{\sigma_{\theta}(S^\lu_{\Theta})}{\sigma_{\theta}(S_{\theta}^g)}
\ge
\frac{\sigma_{\theta}(S^\lu_{\Theta})}{\sigma_{\theta}(S^{*}_{\theta})}
\ge
\min_{\theta' \in \Theta} \frac{\sigma_{\theta'}(S^\lu_{\Theta})}{\sigma_{\theta'}(S^{*}_{\theta'})}
=
g(\Theta, S^\lu_{\Theta}) \mbox{.}
\end{equation*}


The next question is how to find a $\theta=(\theta_e)_{e\in E}\in \Theta$
to make the upper bound
$\bar{\alpha}(\Theta,\theta)$ as small as possible.
In our experiments, we use the following two heuristics and take their minimum.

The first heuristic borrows the intuition from Example~\ref{exp:tight}, which
says that the gap ratio $\alpha(\Theta)$ is close to the robust ratio
$g(\Theta, S^\lu_{\Theta})$ when (a) there are two disjoint seed sets with
similar influence spead, (b) their cascade overlap is small, and
(c) the reachable edges from one seed set use lower end parameters values while the reachable edges from the other seed set use upper end parameters.
Thus in our heuristic, we use PMIA algorithm to find another seed set $S'$
of $k$ nodes
when we remove all nodes in $S^\lu_{\Theta}$.
We then do information cascades from both $S^\lu_{\Theta}$ and $S'$ for an
equal number of times.
Finally, for every edge $e$, if it is sampled more in the information cascade with seed set $S^\lu_{\Theta}$ than with $S'$, we set $\theta_e=l_e$, otherwise we set $\theta_e=r_e$.
The second heuristic is a variant of the first one, where we run a number of
information cascades from $S^\lu_{\Theta}$, and for any edge $e$
that is sampled in at least $10\%$ of cascades, we set $\theta_e=l_e$,
otherwise we set $\theta_e=r_e$.

Other more sophisticated heuristics are possible, but it could be a separate
research topic to find tighter upper bound for the robust ratio, and thus
we only use the simple combination of the above two in this
paper, which is already indicative.
We henceforth use $\bar{\alpha}(\Theta)$ to represent
the upper bound found by the minimum of the above two heuristics.

\subsection{Results}

\subsubsection{$\alpha(\Theta)$ and $\bar{\alpha}(\Theta)$
	with Predetermined Intervals}

In the first experiment we explore the relationship between
the width of confidence interval $\Theta=\times_{e\in E}[l_e,r_e]$ and $\alpha(\Theta)$ together
with $\bar{\alpha}(\Theta)$.
For a given interval width $W$,
we set $l_e=\min\{p_e-\frac{W}{2},0\},r_e=\max\{p_e+\frac{W}{2},1\}$
$\forall e\in E$, where $p_e$ is the ground-truth probability of $e$.
Then we calculate $\alpha(\Theta)$ and $\bar{\alpha}(\Theta)$.
We vary the width $W$ to see the trend of changes of
$\alpha(\Theta)$ and $\bar{\alpha}(\Theta)$.
Figure~\ref{fig1} reports the result on the three graphs with seed set size $k=50$.

\begin{figure}[t]
	\centering
	\includegraphics[scale=0.45]{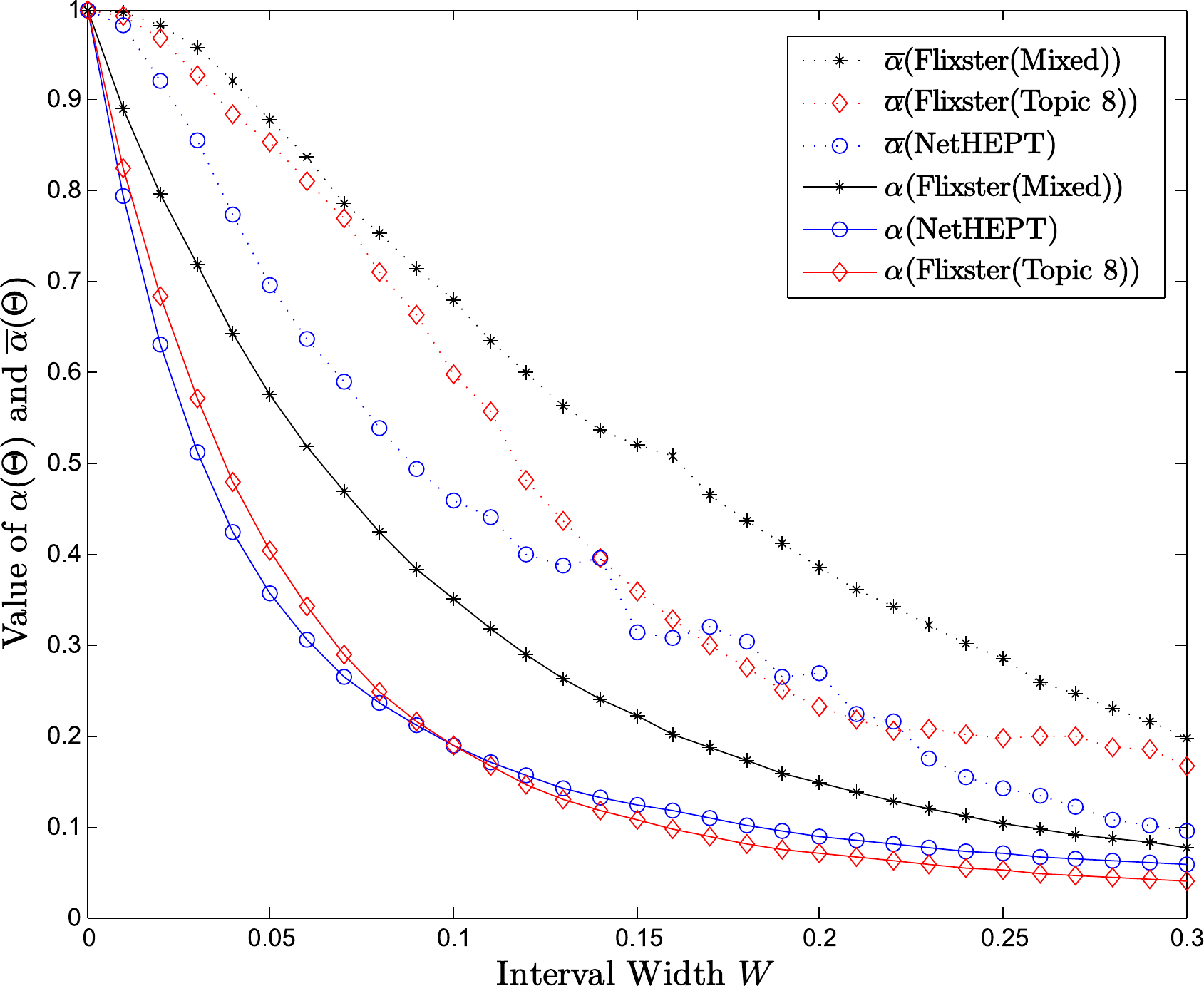}
	\caption{$\alpha(\Theta)$ and $\bar{\alpha}(\Theta)$
		for different widths of confidence interval $W$.}
	\label{fig1}
\end{figure}

First, we observe that as the parameter space $\Theta$ becomes wider,
the value of both $\alpha(\Theta)$ and $\bar{\alpha}(\Theta)$ become smaller,
which matches our intuition that larger uncertainty results in worse
robustness.
Second, there is a sharp decrease of $\alpha(\Theta)$
between $W\in [0,0.1]$ and a much slower decrease afterwards
for all three graphs.
The decrease of $\bar{\alpha}(\Theta)$ is not as sharp as that of $\alpha(\Theta)$
but the decrease also slows down with larger $W$ after $0.2$.
The overall trend of $\alpha(\Theta)$ and $\bar{\alpha}(\Theta)$
suggests that the robust ratio may be sensitive
to the uncertainty of the parameter space, and only when the uncertainty of
the parameter space reduces to a certain level that we can obtain reasonable
guarantee on the robustness of our solution.

As a comparison, we know that the average number
of samples on each edge is $318$ for the learned probabilities in the
Flixster dataset.
This corresponds to an average interval width of
0.293 for topic 8 and 0.265 for the mixed topic.
At these interval widths, $\alpha(\Theta)$ values are approximately
$0.04$ and $0.08$ respectively for the two graphs, and
$\bar{\alpha}(\Theta)$ are approximately $0.12$ and $0.2$ respectively.
This means that, even considering the upper bound  $\bar{\alpha}(\Theta)$,
the robust ratio is pretty low, and thus the learned probabilities
reported in~\cite{barbieri2013topic} may result in quite poor performance for
robust influence maximization.

Of course, our result of $\alpha(\Theta)$ and $\bar{\alpha}(\Theta)$ is only
targeted at the robustness of our {\lugreedy} algorithm, and there could
exist better algorithm having higher robustness performance at the same
uncertainty level. 
Finding a better RIM algorithm seems to be a difficult task, and
we hope that our study could motivate more research in searching for such better
RIM algorithms.
Besides $S^\lu_{\Theta}$, we also independently test the classical
greedy seed set $S^g_{\theta}$ for $\theta=(p_e)_{e\in E}$
on the lower parameter vector $\theta^-$
(that is $\frac{\sigma_{\theta^-}(S^g_{\theta})}{\sigma_{\theta^+}(S^g_{\theta^+})}$ versus $\alpha(\Theta)$),
and the average performance on each data point is $2.45\%$, $1.05\%$, $6.11\%$ worse than $S^\lu_{\Theta}$
for Flixster(Mixed), Flixster(Topic~8) and NetHEPT, respectively.
Therefore, it shows that $S^\lu_{\Theta}$ outperforms $\sigma^g_{\theta}$ in the worse-case scenario, 
and henceforth we only use $S^\lu_{\Theta}$ in the following experiments.


\subsubsection{Results for Sampling algorithms}

Figures~\ref{fig2}, \ref{fig3} and \ref{fig4} reports the result of
$\alpha=\alpha(\Theta)$ and $\bar{\alpha}=\bar{\alpha}(\Theta)$ 
for the three tested graphs respectively, 
when the average number of samples per edge increases. 
For better presentation, we trim all figures as long as 
$\alpha(\text{\USRIM}) = 0.7$. 
(For example, in Flixster(Topic 8), 
{\USRIM} requires $77318$ samples in average for $\alpha$ to reach $0.8$, 
while {\ICSRIM} only needs $33033$, and for {\OESRIM} $\alpha$ sticks to $0.118$.)

For the sampling algorithms, after the $i$-{th} iteration, the observation
set is updated from $\mathcal{M}_{i-1}$ to $\mathcal{M}_i$,
and the average number of samples per edge in the network is calculated.
Markers on each curve in these figures represent the result after one
iteration of the corresponding sampling algorithm.

\begin{figure}[t]
	\centering
	\includegraphics[scale=0.45]{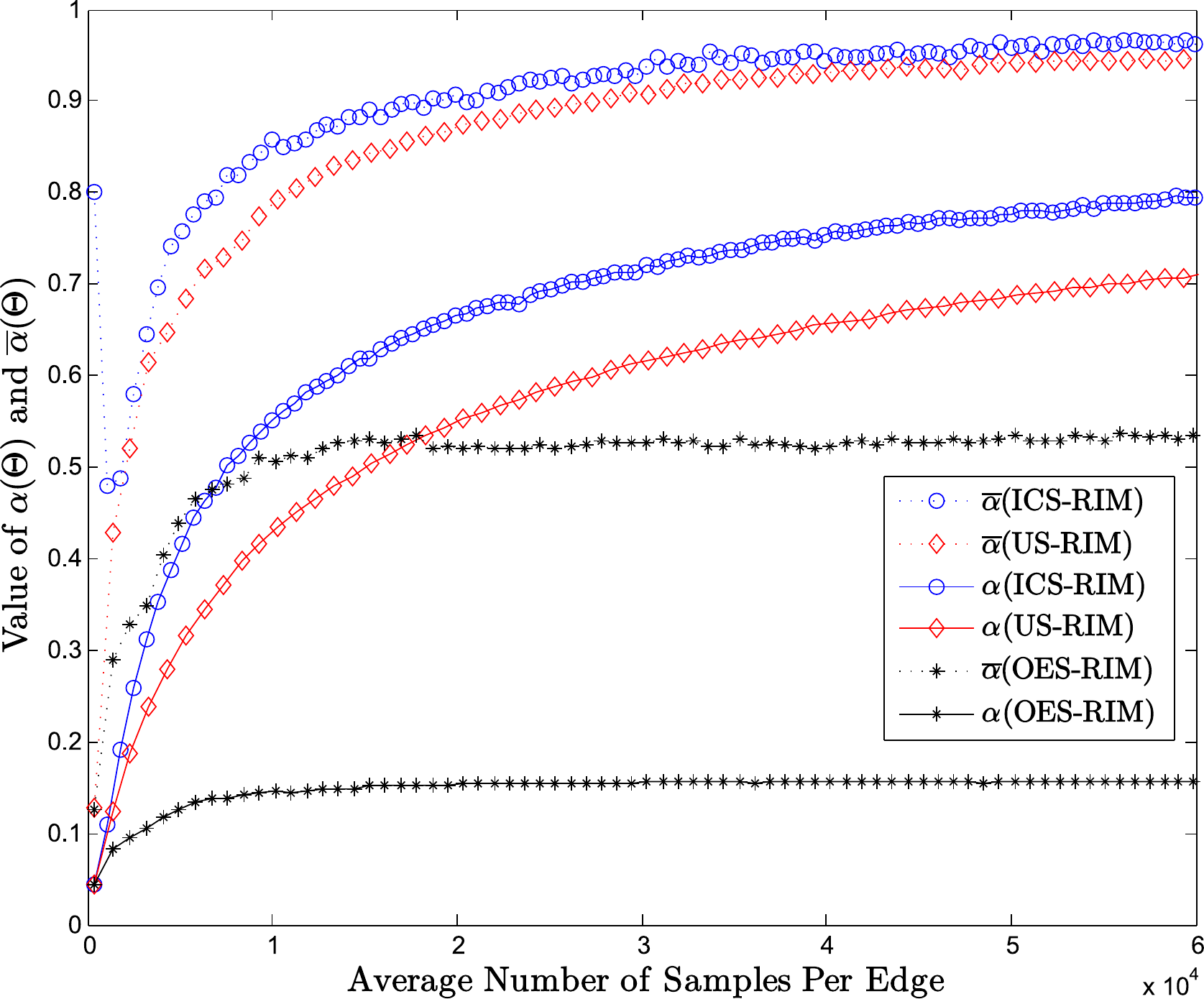}
	\caption{$\alpha(\Theta)$ and $\bar{\alpha}(\Theta)$
		for different average number of samples per edge on graph NetHEPT.}
	\label{fig2}
\end{figure}

\begin{figure}[t]
	\centering
	\includegraphics[scale=0.45]{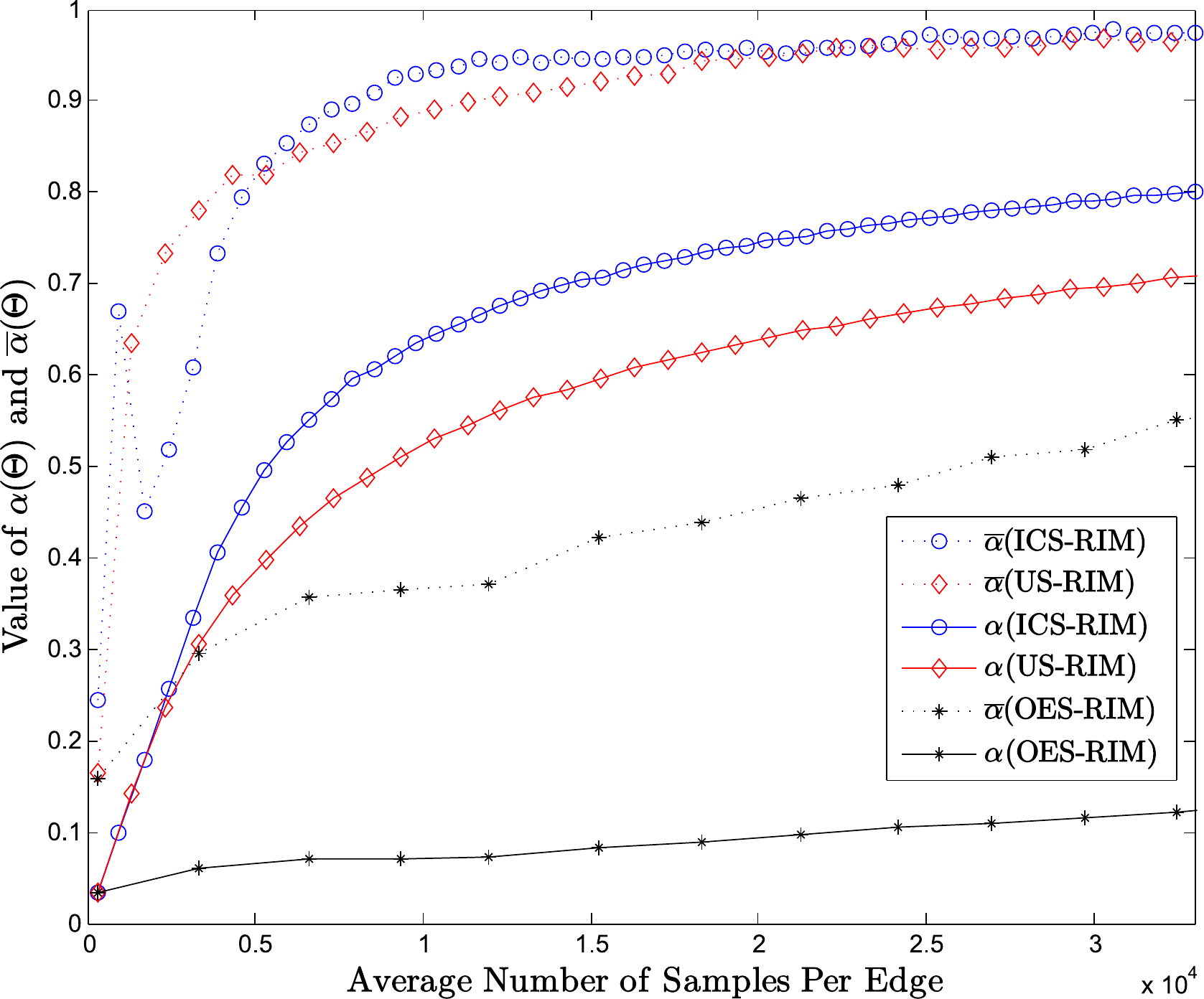}
	\caption{$\alpha(\Theta)$ and $\bar{\alpha}(\Theta)$
		for different average number of samples per edge on graph Flixster(Topic~8).}
	\label{fig3}
\end{figure}

\begin{figure}[t]
	\centering
	\includegraphics[scale=0.45]{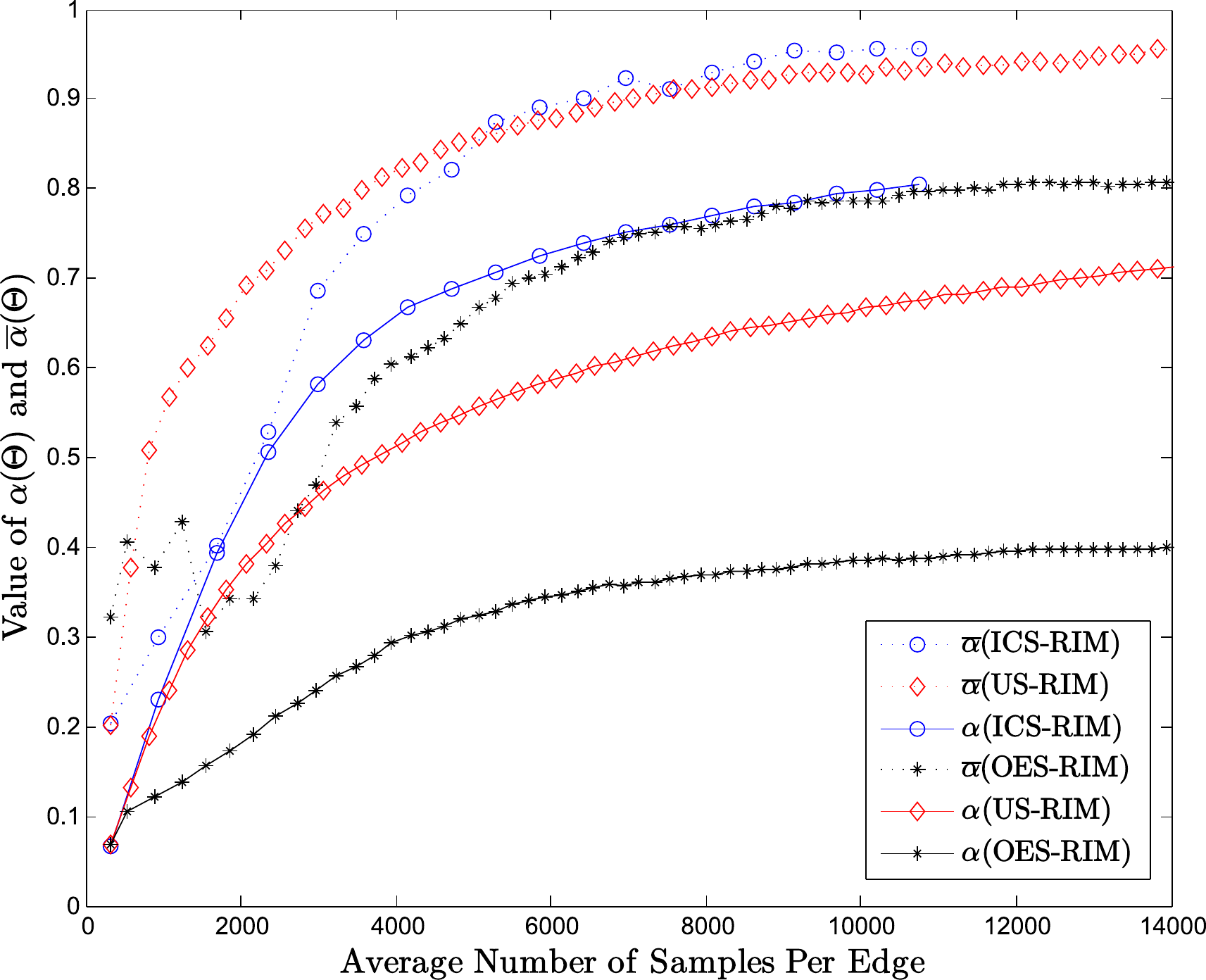}
	\caption{$\alpha(\Theta)$ and $\bar{\alpha}(\Theta)$
		for different average number of samples per edge on graph Flixster(Mixed).}
	\label{fig4}
\end{figure}

The results on all three graphs are consistent.
First, for each pair of $\alpha(\Theta)$ and $\bar{\alpha}(\Theta)$, even though
there is still some gap, 
indicating either the lower bound
or the upper bound may not be tight yet, the trends on
both $\alpha(\Theta)$ and $\bar{\alpha}(\Theta)$ are consistent:
Both increase with the number of samples, even with similar slopes at each
point; and among different
algorithms, the ranking order and relative change are consistent
with both $\alpha(\Theta)$ and $\bar{\alpha}(\Theta)$.
All these consistency suggests that gap ratio $\alpha(\Theta)$ could be used
as an indicator for the robustness of Algorithm {\lugreedy}, and
it is reasonable to use $\alpha(\Theta)$ in comparing the performance
of different algorithms.

Second, comparing the performance of three algorithms, we see that
both {\USRIM} and {\ICSRIM} are helpful in improving the robust ratio of
the selected seed set, and {\ICSRIM} is better than {\USRIM}, especially
when the sample size increases.
The baseline algorithm {\OESRIM}, however, performs significantly poorer
than the other two, even though it is also an adaptive algorithm as
{\ICSRIM}.
The reason is that, the lower-upper greedy set $S^{\lu}_{\Theta}$ changes little
after a certain number of iterations in {\OESRIM},
and thus only a small number of edges (out edges of $S^{\lu}_{\Theta}$)
are repeatedly sampled.
The probabilities on these edges are already estimated very accurately
while other edge probabilities are far from accurate.
It is the inaccurate edges that make $\alpha(\Theta)$ and the best robust ratio small.
In contrast, {\ICSRIM} uses information cascades to sample
not only edges directly connecting to the seed set but also edges
that can be potentially reached.
This suggests that it is important for a sampling method to balance the
sampling between critical edges and other potentially useful edges
in order to achieve better robustness in influence maximization.

Overall, the results suggest that information cascade based sampling method
stands out as a competitive choice when we can adaptively sample more
edges to achieve better robustness.
If adaptive sampling is not possible, predetermined uniform sampling
may also perform reasonably well.

\section{Conclusion}  \label{sect:conclusion}

In this paper, we propose the study of robust influence maximization to address
the impact of uncertainty in edge probability estimates that would inevitably occur
in practice to the influence maximization task.
We propose the {\lugreedy} algorithm with a proven solution-dependent bound,
and further propose sampling methods, in particular information cascade
based adaptive sample method to effectively reduce the uncertainty and
increase the robustness of the {\lugreedy} algorithm.
The experimental results validate the usefulness of the {\lugreedy} algorithm
and the information cascade based sampling method {\ICSRIM}.
Moreover, the results indicate that robustness may be sensitive to the uncertainty
of parameter space, and learning algorithms may need more data to achieve
accurate learning results for robust influence maximization.

Our work opens up a number of research directions.
First, it is unclear what could be the upper bound of the best robust ratio given
an actual network and learned parameter space.
Answering this question would help us to understand whether robust
influence maximization is intrinsically
difficult for a particular network or it is just our algorithm that does not
perform well.
If it is the latter case, then an important direction is to
design better robust influence maximization algorithms.
Another direction is how to improve sampling methods and learning methods
to achieve more accurate parameter learning, which seems to be crucial
for robust influence maximization.
In summary, our work indicates a big data challenge on social influence research
--- the data on social influence analysis is still not big enough,
such that the uncertainty level in model learning may result in
poor performance for influence maximization.
We hope that our work could encourage further researches to meet this
challenge from multiple aspects including
data collection, data analysis, and algorithm design.

\section*{Acknowledgment} \label{sect:acknowledge}
The research of Wei Chen is partially supported by the National
Natural Science Foundation of China (Grant No. 61433014).

\bibliographystyle{abbrv}
\bibliography{sigproc}

\appendix

\section{Proof of Theorem~3} 

\begin{proof}

{\flushleft \em (Case~\ref{pro:3-ratios-1}): }
Let $G$ be an $n$-clique and $\Theta = \times_{e\in E} [0, 1]$, i.e., for every edge $e$, $l_e=0$ and $r_e=1$.
For arbitrary set $S=\{v_1,\cdots, v_k\}$, there exists a valid parameter vector $\theta = (p_e)_{e \in E} \in \Theta$, where $p_e=0$ for all $E_S=\{e=(u,v)\mid u\in S \text{ or } v\in S\}$ and $p_e=1$ for all $e\notin E_S$.
Then, $\sigma_{\theta}(S)=k$ and $\sigma_{\theta}(S^{*}_{\theta})=n-1$,
which implies that
$g(\Theta, S) =
	\min_{\theta \in \Theta} \frac{\sigma_{\theta}(S)}{\sigma_{\theta}(S^{*}_{\theta})}
	\leq \frac{k}{n-1}
$.
For any set $S$ of size $k$, the above holds, thus we can conclude that
\[
	\max_{|S|=k}g(\Theta,S) = O\left(\frac{k}{n}\right) \mbox{.}
\]

{\flushleft \em (Case~\ref{pro:3-ratios-2}): }
Consider graph $G=(V,E)$ such that $V=A\cup B, |A|=|B|=\frac{n}{2}$ and $E=\{(u,v)\mid u,v\in A \text{ or }u,v\in B\}$, and let $E(A)$ be the set of edges with two endpoints in $A$ and $E(B)$ defined similarly. The problem is to find a single seed ($k=1$) such that the influence spread is maximized. Let $p=\frac{2}{n}$ and the input instance is $l_e=p-\epsilon$ and $r_e=p+\epsilon$ for every edge $e$ such that $[l_e,r_e]$ covers the critical interval of Erd\H{o}s-R\'{e}nyi random graph with $\frac{n}{2}$ nodes.

Now since every node is seemingly the same for any algorithm, suppose the algorithm chooses a seed $u\in A$,
then consider the worst-case $\theta$ where for every $e\in E(A)$, $p_e=l_e$ and for every $e\in E(B)$, $p_e=r_e$.
It can be figured out that the optimal solution is an arbitrary node $v\in B$.
Since $\sigma_{\theta}(\{u\}) = O(\log n)$ and $\sigma_{\theta}(\{v\})=\Theta(n)$, then the ratio $r = O(\frac{\log n}{n})$.

{\flushleft \em (Case~\ref{pro:3-ratios-3}): }
Consider graph $G=(V,E)$ such that
	$V$ is composed of disjoint sets $A_1, A_2,\ldots, A_{\sqrt{n}}$ where each $|A_i|=\sqrt{n}$,
and
	$E=\{(u,v)\mid u,v\in A_i, \forall i=1,\cdots,\sqrt{n}\}$.
Let $E(A_i)$ be the set of edges with two endpoints in $A_i$.
The problem is to find a single seed ($k=1$) such that the influence spread is maximized.
Let $p=\frac{1}{\sqrt{n}}$, and the input instance is $l_e=p-\epsilon$ and $r_e=p+\epsilon$ for every edge $e$
such that $[l_e,r_e]$ covers the critical interval of Erd\H{o}s-R\'{e}nyi random graph with $\sqrt{n}$ nodes.
Now every node appears to be symmetric from the input.

Denote $q_{i}$ as the probability of choosing a node in $A_{i}$.
Consider any distribution assigned on
$A_{1}, A_{2},\ldots,A_{\sqrt{n}}$, i.e. $q_1+q_2+\cdots+q_{\sqrt{n}}=1$,
and let the random seed set be $\tilde{S}$.
Without loss of generality, let $q_1$ be the smallest one. Then consider the worst-case $\theta$ where for every $e\in E(A_1)$, $p_e=r_e$ and for every $e\in E(A_i),i\ge 2$, $p_e=l_e$. It is obvious that the optimal solution $S_{\theta}^{*}$ is an arbitrary point $v\in A_1$.
Since
\[
	\E\left[ \sigma_{\theta}(\tilde{S})\right] \le \frac{1}{\sqrt{n}} \cdot \sqrt{n} + \left( 1-\frac{1}{\sqrt{n}} \right) O(\log \sqrt{n})=O(\log n) \mbox{,}
\]
and
\[
	\sigma_{\theta}(S_{\theta}^{*})=\Theta(\sqrt{n}) \mbox{,}
\]
which completes the proof.
\end{proof}

\section{Proof of Lemmas}

\begin{proof}[(Lemma~\ref{lem:mul})]
Since when $\sigma_{\theta}(S)$ is regarded as a function on $\theta$ (if $S$ is fixed), it is monotonically increasing,
thus it suffices to consider the case that $\forall e\in E$, $r_e=(1+\lambda)l_e$.


Flipping coins for every edge according to the probability parameter $\theta$, and we have a live-edge (random) graph $L$.
Let $E(L)$ denote the set of edges in $L$, and $\Pr_{\theta}[L]$ be the probability yielding $L$.
We use $R_L(S)$ to denote the reachable set from $S$ in $L$. Then, the influence spread function has a linear form as follows,
\[
\sigma_{\theta}(S)=\sum_L\Pr_{\theta}[L]\cdot|R_L(S)| \mbox{.}
\]

As a convention, for any edge $e\in E$, we denote conditional probability
$\Pr_{\theta}[L|e] = \Pr_{\theta}\left[L|e\in E(L)\right]$,
and
$\Pr_{\theta}[L|\bar{e}] = \Pr_{\theta}\left[L|e\notin E(L)\right]$.
Then, we have
\begin{align*}
&\frac{\sigma_{\theta^+}(S)}{\sigma_{\theta^-}(S)} \\
=&\frac{\displaystyle\sum_{L:e\in E(L)}r_e |R_L(S)| \Pr_{\theta^{+}}{[L|e]} + \sum_{L:e\not\in E(L)}(1 - r_e) |R_L(S)|\Pr_{\theta^{+}}{[L|\bar{e}]}}
	{\displaystyle\sum_{L:e\in E(L)}l_e |R_L(S)| \Pr_{\theta^{-}}{[L|e]} + \sum_{L:e\not\in E(L)}(1 - l_e) |R_L(S)| \Pr_{\theta^{-}}{[L|\bar{e}]}} \mbox{.}
\end{align*}

When we fixed $l_{e'}$ for all $e'\not=e$, we have
\[
\frac{\sigma_{\theta^+}(S)}{\sigma_{\theta^-}(S)}=\frac{Al_e+B}{Cl_e+D} \mbox{,}
\]
where $A,B,C,D$ are not dependent on $l_e$. It can be observed that the ratio is monotone with $l_e$, and is thus maximized either when $l_e=0$ or when $l_e=\frac{1}{1+\lambda}$.

Similar analysis for other edges, we can conclude that when the ratio is maximized, it must holds that $\forall e\in E$, $l_e=0$ or $l_e=\frac{1}{1+\lambda}$.
Since when $l_e=0$, it holds that $r_e=0$, thus we can just delete this edge from the graph.
Delete all such edges, and it ends up with a graph $G_1=(V,E_1)$ such that the probability interval on every edge is $[\frac{1}{1+\lambda},1]$.
And it can be seen that $R_{G_1}(S)$ is determined when probability on all edges are $1$.


Given set $S$, denote the influence spread for any graph $G$ under any parameter vector $\theta$ as $\sigma_{\theta}^{G}(S)$ explicitly.
If there exists a directed cycle $v_0\to v_1\to\cdots \to v_i\to v_0$ in graph $G_1$. 
Then it can be seen that either all nodes in this cycle is in $R_{G_1}(S)$, or none of them is in. 
In both cases, we can remove some edge (e.g. $v_i\to v_0$) from $E_1$ and obtain a new graph $G_2$ (e.g. $G_2=(V,E_1 \setminus \{(v_i,v_0)\})$) such that $\sigma_{\theta^+}^{G_1}(S)=\sigma_{\theta^+}^{G_2}(S)$ while $\sigma_{\theta^-}^{G_1}(S)\geq\sigma_{\theta^-}^{G_2}(S)$. Thus,
\[
\frac{\sigma_{\theta^+}^{G_1}(S)}{\sigma_{\theta^-}^{G_1}(S)}\leq \frac{\sigma_{\theta^+}^{G_2}(S)}{\sigma_{\theta^-}^{G_2}(S)} \mbox{.}
\]

Removing can be done since if none of the nodes are in $R_{G_1}(S)$, then deleting one edge will not change either $\sigma_{\theta^+}^{G_1}(S)$ or $\sigma_{\theta^-}^{G_1}(S)$, and if all of the nodes are in, then there must exists $v_0$ in the cycle such that $v_0\in S$ or $v_0$ can be reached from a path directing from some node (in $S$) outside the cycle to it, then deleting the edge $(v_p,v_0)$ can be proved to satisfy the above property.

Repeat deleting edges until the remaining graph is a directed acyclic graph (DAG), denoted by $G'$. Then it can be split into finite subgraphs $T_1,T_2,\ldots,T_j$
where each $T_i$ is a connected DAG, and it is immediate that
\[
\frac{\sigma_{\theta^+}^{G'}(S)}{\sigma_{\theta^-}^{G'}(S)}\le \max_{1\le i\le j}\frac{\sigma_{\theta^+}^{T_i}(S)}{\sigma_{\theta^-}^{T_i}(S)} \mbox{.}
\]

It remains to analyze the ratio in a connected DAG $T_i$, and we need more notations before that. First, the DAG $T_i$ naturally induces a topological order on nodes (we can therefore call the nodes in $T_i$ be $V(T_i):=\{v_1,\cdots,v_{|T_i|}\}$), in which every edge in $E(T_i)$ is directing from a node with smaller order to a larger order. Let $S_i=S\cap V(T_i)$, and let $R(S_i)$ be the subset of nodes in $V(T_i)$ that is reachable with positive probability (therefore $R(S_i)$ naturally contains nodes in $S_i$). Besides, for any $v\notin S_i$, let $d(S_i,v)$ denotes the length of shortest path directing from some node in $S_i$, and for any $v\in S_i$, define $d(S_i,v)=0$. Thus,
\[
\sigma_{\theta^+}^{T_i}(S_i)=|R(S_i)| \mbox{.}
\]

Let $\beta=\frac{1}{1+\lambda}$. For any path of length $l \geq 0$ from $S_i$ to $v$, the activating probability of
that path is $\beta^l$ under $\theta^-$. Then, we have
\[
\begin{aligned}
\sigma_{\theta^{-}}^{T_i}(S_i)&=|S_i|+\sum_{v\in V(T_i) \setminus S_i}\Pr\left[\text{v is reached}\right]\\
&\geq |S_i|+\sum_{v\in R(S_i) \setminus  S_i}\beta^{d(S_i,v)}\\
&\geq \sum_{v \in R(S_i)} \beta^{d(S_i,v)}\\
&\geq |R(S_i)| \beta^{|T_i|} \mbox{.}
\end{aligned}
\]

Therefore,
\[
\frac{\sigma_{\theta^+}^{G}(S)}{\sigma_{\theta^-}^{G}(S)}
\leq \frac{\sigma_{\theta^+}^{G'}(S)}{\sigma_{\theta^-}^{G'}(S)}
\leq \max_{1\le i\le p}\frac{\sigma_{\theta^+}^{T_i}(S)}{\sigma_{\theta^-}^{T_i}(S)}
\leq \max_{1\le i\le p} \beta^{-|T_i|}
\leq (1+\lambda)^{n} \mbox{.}
\]

To prove the second inequality of this lemma, by definition we have
	$S_{\theta}^{*}=\arg\max_{|S|\leq k}\sigma_{\theta}(S)$.
Note that for all $\theta\in\Theta$,
	$\sigma_{\theta^{+}}(S_{\theta^{+}}^{*})\geq\sigma_{\theta}(S_{\theta}^{*})$.
Then we have
\begin{align*}
\max_{|S|\leq k}\min_{\theta\in\Theta}\frac{\sigma_{\theta}(S)}{\sigma_{\theta}(S_{\theta}^{*})}&\geq
\max_{|S|\leq k}\min_{\theta\in\Theta}\frac{\sigma_{\theta}(S)}{\sigma_{\theta^{+}}(S_{\theta^{+}}^{*})}\\
&=\frac{\sigma_{\theta^{-}}(S_{\theta^{-}}^{*})}{\sigma_{\theta^{+}}(S_{\theta^{+}}^{*})}\\
&\geq \frac{\sigma_{\theta^{-}}(S_{\theta^{+}}^{*})}{\sigma_{\theta^{+}}(S_{\theta^{+}}^{*})}\\
&\geq \min_{|S|\subseteq V}\frac{\sigma_{\theta^{-}}(S)}{\sigma_{\theta^{+}}(S)}\\
&\geq \frac{1}{(1+\lambda)^{n}} \mbox{.}
\end{align*}

This completes the proof for Lemma~\ref{lem:mul}.
\end{proof}

\begin{proof}[(Lemma~\ref{lem:conf})]

First, we focus on one fixed edge $e$. According to Chernoff bound, we have
\[ 
	\Pr[|\hat{p}_e - p_e|\le p_e\delta] \ge 1 - 2e^{-\frac{1}{3}\delta^2p_et_e}. 
\]
Let $\gamma = 2me^{-\frac{1}{3}\delta^2p_et_e}$, and
	$\delta = \sqrt{\frac{3}{t_e}\ln\frac{2m}{\gamma}}\frac{1}{\sqrt{p_e}} = \frac{c_e}{\sqrt{p_e}}$.
Then, with probability no less than $1 - \frac{\gamma}{m}$,
we see that $p_e$ should satisfy the constraint
\[ 
	|\hat{p}_e - p_e|\le c_e\sqrt{p_e}, 
\]
thus we have
\[ 
	\hat{p}_e + \frac{c_e^2}{2} - c_e\sqrt{\frac{c_e^2}{4} + \hat{p}_e}\le p_e \le \hat{p}_e + \frac{c_e^2}{2} + c_e\sqrt{\frac{c_e^2}{4} + \hat{p}_e}. 
\]
By definition of $l_e$, $r_e$ and the fact that $p_e \in [0,1]$,
therefore we have
\[
	\Pr[l_e \le p_e \le r_e] \ge 1 - \frac{\gamma}{m}.
\]
By union bound, we can conclude that
\[
	\Pr[l_e \le p_e \le r_e,\, \forall e\in E] \ge 1 - \gamma,
\]
which completes the proof of Lemma~\ref{lem:conf}.
\end{proof}


\section{Proof of Theorem~6} 

\begin{proof}

{\flushleft \em Setting~\ref{thm-add-case}:}
First, since every $e$ is probed for $t=\frac{2m^2n^2 \ln \frac{2m}{\gamma}}{k^2\epsilon^2}$ times, using the additive form of Chernoff-Hoeffding Inequality we have
\[
\Pr\left[ \abs{\frac{1}{t}\sum_{i=1}^{t}x^i_e-p_e} >\frac{k\epsilon}{2mn} \right] \le 2\exp\left(-\frac{k^2\epsilon^2}{2m^2n^2}\cdot t\right)\le \frac{\gamma}{m} \mbox{.}
\]
Then by union bound, it holds that
\[
\Pr\left[\forall e\in E, \abs{\frac{1}{t}\sum_{i=1}^{t} x^i_e-p_e} >\frac{k\epsilon}{2mn}\right]\le \gamma \mbox{.}
\]

For every $e\in E$, we set $l_e=\frac{1}{t}\sum_{i=1}^{t}x^i_e-\frac{k\epsilon}{2mn}$, and $r_e=\frac{1}{t}\sum_{i=1}^{t}x^i_e+\frac{k\epsilon}{2mn}$, then with probability $\ge 1-\gamma$, it holds that $\theta\in \Theta$.

Therefore, for every $S$, according to Lemma~\ref{lem:add},
\begin{align*}
\frac{\sigma_{\theta^{-}}(S)}{\sigma_{\theta^{+}}(S)}&
\ge 1-\frac{\sigma_{\theta^{+}}(S)-\sigma_{\theta^{-}}(S)}{\sigma_{\theta^{+}}(S)}\\
&\ge 1-\frac{mn\cdot\frac{k\epsilon}{mn}}{k}\\
&=1-\epsilon \mbox{.}
\end{align*}
Thus, $\frac{\sigma_{\theta^{-}}(S^g_{\theta^+})}{\sigma_{\theta^{+}}(S^g_{\theta^+})} \geq 1-\epsilon$ 
also holds.

Now, since we use $S^\lu_\Theta$ as the solution, applying \Cref{thm:main}, we have
\begin{align*}
g(\Theta,S^\lu_\Theta) 
	\ge \alpha(\Theta)\left(1-\frac{1}{e}\right)
	= \frac{\sigma_{\theta^{-}}(S^\lu_\Theta)}{\sigma_{\theta^{+}}(S_{\theta^{+}}^g)} \left(1-\frac{1}{e}\right) \\
	\ge \frac{\sigma_{\theta^{-}}(S^g_{\theta^+})}{\sigma_{\theta^{+}}(S^g_{\theta^+})}  \left(1-\frac{1}{e}\right)
	\ge \left( 1-\epsilon \right) \left(1-\frac{1}{e}\right),
\end{align*}
where the second inequality holds due to $\sigma_{\theta^{-}}(S^\lu_\Theta) \ge \sigma_{\theta^{-}}(S^g_{\theta^+})$
 by definition of \eqref{def:lu-greedy-solution}.

{\flushleft \em Setting~\ref{thm-mul-case}:}
%
%
%
%
Denote $a = \frac{\ln \frac{1}{1-\epsilon}}{2n + \ln \frac{1}{1-\epsilon}}$ for convenience.
Since every edge $e$ is probed for $t=\frac{3 \ln \frac{2m}{\gamma} }{p a^2} \geq \frac{3 \ln \frac{2m}{\gamma} }{p_e a^2}$ times,
the probability of upper and lower tails derived by the multiplicative form of Chernoff-Hoeffding Inequality is
\begin{align*}
& \Pr\left[\frac{1}{t}\sum_{i=1}^{t}x^i_e \ge (1+a)p_e\right]\le e^{-\frac{a^2}{3}\cdot p_e t} \le \frac{\gamma}{2m}
\\
& \Pr\left[\frac{1}{t}\sum_{i=1}^{t}x^i_e \le (1-a)p_e\right]\le e^{-\frac{a^2}{3}\cdot p_e t} \le \frac{\gamma}{2m} \mbox{.}
\end{align*}
Then by union bound, it holds that
\[
	\Pr\left[\forall e\in E, \frac{1}{1+a}\frac{\sum_{i=1}^{t} x^i_e}{t} \le p_e \le \frac{1}{1-a} \frac{\sum_{i=1}^{t} x^i_e}{t} \right] \ge 1 - \gamma \mbox{.}
\]

Now suppose the above bound is satisfied.
For every edge $e \in E$, let $r_e = (1+a) p_e$ and $l_e = (1-a) p_e$.
Then, we have
$r_e \leq \frac{1+a}{1-a}  \frac{\sum_{i=1}^{t}x^i_e}{t} \le (1 + \frac{1}{n} \ln \frac{1}{1-\epsilon}) \frac{\sum_{i=1}^{t}x^i_e}{t}$.
On the other hand, it is easy to check that $r_e = \frac{1+a}{1-a} l_e \le (1 + \frac{1}{n} \ln \frac{1}{1-\epsilon}) l_e$.
According to Lemma~\ref{lem:mul}, for any set $S$,
\begin{align*}
\frac{\sigma_{\theta^{-}}(S)}{\sigma_{\theta^{+}}(S)}&
\ge \left( 1+\frac{1}{n} \ln\frac{1}{1-\epsilon} \right)^{-n}
\ge 1-\epsilon \mbox{.}
\end{align*}
Thus, $\frac{\sigma_{\theta^{-}}(S^g_{\theta^+})}{\sigma_{\theta^{+}}(S^g_{\theta^+})} \geq 1-\epsilon$ 
also holds. Similar to Setting~\ref{thm-add-case}, 
then we can apply Theorem~\ref{thm:main} to derive the theorem.
\end{proof}

\end{document}